\documentclass[10pt, onefignum, onetabnum]{siamonline190516}
\usepackage{amsmath}
\usepackage{amsfonts}
\usepackage{amssymb}
\usepackage{mathtools}
\usepackage{enumitem}
\usepackage{bm}
\usepackage{bbm}
\usepackage{cleveref}
\usepackage{graphicx}
\usepackage{subcaption}
\usepackage{caption}
\usepackage{tikz}
\usepackage{hyperref}
\usepackage{cleveref}
\usepackage{wrapfig}
\usepackage{enumitem}
\usepackage{algorithm}
\usepackage{algorithmicx}
\usepackage[noend]{algpseudocode}
\usepackage{setspace}
\usepackage{microtype}
\usepackage{amscd}
\usepackage{mathrsfs}
\usepackage{stmaryrd}
\usepackage{tensor}
\usepackage{braket}
\usepackage{xfrac}
\usepackage{array}
\usepackage{booktabs}
\usepackage{siunitx}
\usepackage{cancel}
\usepackage{thmtools}
\usepackage{mdframed}
\usepackage{pgfplots}
\usepackage{tikz-cd}
\usepackage[version=4]{mhchem}
\usepackage{physics}
\usepackage{tikz}
\usepackage{tikz-cd}
\newcommand{\tthanks}[1]{%
  \begingroup
  \renewcommand{\thefootnote}{\dagger}% 
  \footnotemark
  \footnotetext{#1}%
  \endgroup
}
\newsiamremark{remark}{Remark}
\newsiamremark{example}{Example}
\newsiamremark{result}{Result}
\newsiamremark{setting}{Setting}
\newsiamremark{conjecture}{Conjecture}
\newcommand{\cso}{\operatorname{CSO}}
\newcommand{\co}{\operatorname{CO}}
\newcommand{\flip}{\mathfrak{R}} 
\newcommand{\falign}[2]{\sigma(#1\,;#2)}
\newcommand{\falignlim}[2]{\overline{\sigma}(#1\,;#2,\tau)}
\newcommand{\falignlimgen}[3]{\overline{\sigma}(#2\,;#3,#1)}
\newcommand{\faligndata}[2]{\overline{\sigma}(#1\,;\mathcal{X}_N)}
\newcommand{\losslim}[2]{\mathcal{L}(#1\,; #2, \tau)}
\newcommand{\lossdata}[2]{\mathcal{L}(#1 \, ;#2)}
\newcommand{\losslimgen}[3]{\mathcal{L}(#2\,;#3,#1)}
\DeclareMathOperator*{\supp}{supp}
\renewcommand{\bullet}{\begin{tikzpicture}\draw[fill=black](0,0)++(0,0.2em) circle (1.2 pt);\end{tikzpicture}}
\DeclareMathOperator*{\argmin}{arg\,min}
\DeclareMathOperator*{\argmax}{arg\,max}

\headers{A New Approach for Multireference Alignment}{Vahid Shahverdi, Emanuel Ström, and Joakim And\'en}

\title{Moment Constraints and Phase Recovery \\
        for Multireference Alignment}

\author{Vahid Shahverdi\tthanks{Department of Mathematics, KTH Royal Institute of Technology, Stockholm, Sweden 
(\email{vahidsha@kth.se}, \email{emastr@kth.se}, \email{janden@kth.se}).}
\and Emanuel Ström$^\dagger$
\and Joakim Andén$^\dagger$}

\date{}     
\begin{document}
\maketitle
\begin{abstract}
\bfseries
{Multireference alignment (MRA) refers to the problem of recovering a signal from noisy samples subject to random circular shifts. Expectation--maximization (EM) and variational approaches use statistical modeling to achieve high accuracy at the cost of solving computationally expensive optimization problems. The method of moments, instead, achieves fast reconstructions by utilizing the power spectrum and bispectrum to determine the signal up to shift. Our approach combines the two philosophies by viewing the power spectrum as a manifold on which to constrain the signal. We then maximize the data likelihood function on this manifold with a gradient-based approach to estimate the true signal. Algorithmically, our method involves iterating between template alignment and projections onto the manifold. The method offers increased speed compared to EM and demonstrates improved accuracy over bispectrum-based methods.}
\end{abstract}


\section{Introduction}
\label{sec:intro}
\emph{Multireference alignment} (MRA) is a class of inverse problems, in which the task is to estimate a signal from observations (references) that have undergone a random transformation and possibly under measurement noise. The key distinction between MRA as opposed to other inverse problems is that the transformations are invertible, and in particular, they can be represented as random elements of some group acting on the data. 

MRA is often viewed as a simplified version of single-particle cryo-electron microscopy (cryo-EM). In cryo-EM, the objective is to estimate a three-dimensional object e.g., a molecule from two-dimensional noisy projections taken from various (unknown) angles, which correspond to the unknown transformation in MRA. While retaining some of the original complexities of the full cryo-EM problem, MRA omits the lossy projection operator, which allows for a more straightforward analysis of the estimation problem. The fundamental mathematical principles of cryo-EM are explored in \cite{singer2018mathematics,bendory2020single,relion}. Direct applications of MRA also appear in standard cryo-EM image processing tasks, such as class averaging~\cite{sorzano2010clustering, ma2020heterogeneous}.

In this work, we consider specifically the classical one-dimensional MRA problem. Formally, we denote the clean signal by $x \in \mathbb{R}^L$ and let $N$ be the number of observations. The observations are then given by applying the shifts $r_1, r_2, \ldots, r_N \in \{0, 1, \ldots, L-1\}$ to $x$ and adding the noise terms $\epsilon_1, \epsilon_2, \ldots, \epsilon_N$ to give 
\begin{equation}
\label{eq:observ-sig}
\xi_i = \sigma_{r_i}(x) + \epsilon_i,
\end{equation}
where $\sigma_{r_i}$ denotes a circular translation applied to $x$ by $r_i$ units, defined as $\sigma_{r_i}(x)[n] = x[n - r_i]$ (indices of $x$ are here considered modulo $L$). Furthermore, we suppose that the $\epsilon_i$s are iid normal random vectors with variance $\tau^2$. Multireference alignment is a somewhat outdated name, since it implies that a key component of the model is to align the references -- a task that is virtually impossible in high-noise scenarios -- many approaches to MRA therefore try to circumvent alignment altogether~\cite{bendory2017bispectrum}. Broadly speaking, approaches for addressing MRA can be categorized into two types: those that explicitly align signals as a part of the reconstruction, and those that do not. We refer to these as \emph{aligning} and \emph{non-aligning}.

\paragraph{Aligning Methods}
The first type of algorithms focus on approximating the shifts ${r_i}$ for each of the observed signals $\xi_i$, aligning them, and subsequently averaging to obtain an approximation of the true signal up to shift.

One example of an aligning-type method is \emph{template alignment}. This method aligns the observed signals $\xi_i$ against some reference signal $z$, the \emph{template}. The alignment is achieved by averaging the signals $\sigma_{r^*_i}(\xi_i)$, where $r^*_i$ is given by
\begin{equation}
\label{eq:find_shift}
r^*_i = \underset{k}{\operatorname{argmax}} \left(\sum \limits_{m=0}^{L-1} z[m]\xi_i[m+k]\right), \quad i=1,\ldots,N.
\end{equation}
The shift estimates $r^*_i$ can be computed efficiently in $\mathcal{O}(L\log L)$ time using a fast Fourier transform (FFT) and the aligned observations can then be averaged in $\mathcal{O}(NL \log L)$ time. %Despite the speed advantage of this method, template alignment encounters challenges in achieving reliable reconstruction, even under moderate signal-to-noise ratios (SNRs). 

Template alignment is one of the fastest methods for multireference alignment, and is trivial to parallelize. However, it does not consider the alignment quality between pairs of aligned data, only between data and template. As a result, it is sensitive to the choice of template, and produces biased reconstructions even under moderate signal-to-noise ratios (SNR). Some standard approaches to mitigate the template selection bias include iterated template alignment, which requires multiple passes over the data, and more advanced methods such as \cite{Kosir}, which updates the template throughout a single alignment procedure.

Angular synchronization is an aligning method that improves on the shortcomings of template alignment by modeling the optimal relative shifts between all pairs of signals as noisy measurements of the true shifts \cite{singer2011angular}. The maximum likelihood estimate (MLE) of the shifts then corresponds to a least-squares problem on the circle group. Its solution can be efficiently approximated using spectral methods or semidefinite programming. Angular synchronization is more robust to noise compared to template alignment. However, it is more computationally demanding and scales poorly in the number of observations $N$ due to the need to process all observation pairs.

Variational methods have also been used to jointly estimate $x$ and $\{r_i\}_{i=1}^N$ by minimizing the negative log-likelihood of the data $\mathcal{X}_N=\{\xi_i\}_{i=1}^N$. This produces the MLE of $x$:
\begin{equation}
    \min_{x,\{r_i\}_{i=1}^N}-\log p(\mathcal{X}_N\mid x, \{r_i\}_{i=1}^N) = \min_{x,\{r_i\}_{i=1}^N}\sum_{i=1}^N \|\xi_i - \sigma_{r_i}(x)\|^2 + \mathrm{const}\label{eq: mle}
\end{equation}
One can easily determine $x$ given $\{r_i\}_{i=1}^N$, but recovering the shifts is NP-hard. Relaxation methods combined with semidefinite programming can achieve good results \cite{singer2014relaxed}.

\paragraph{Non-aligning Methods}
There are several methods that attempt to circumvent the challenge of shift estimation by recovering the true signal $x$ directly. One of the most successful examples of non-aligning type methods is the expectation--maximization (EM) algorithm. This method aims to calculate the marginalized maximum likelihood estimator~(MMLE) for the signal, wherein marginalization occurs over shifts in~\cref{eq: mle}. This approach mitigates the challenge of alignment by operating not on estimates of shifts directly but on estimates of their probability distributions. EM demonstrates outstanding numerical performance~\cite{bendory2017bispectrum} but is very demanding from a computational perspective, and a good theory describing its convergence behavior is lacking.

Another non-aligning method for tackling this problem is the method of moments, which estimates the moments of the true signal from the observed signals $\xi_i$ and uses these to reconstruct $x$.
Typically, second- and third-order moments are used, which uniquely describe the clean signal up to a circular shift. 
In \cite{bendory2017bispectrum}, it is shown that approximation of these moments from $\xi_i$ yields an estimation of the true signal.
Unfortunately, the computation of the bispectrum is considerably sensitive to noise, so the method is less reliable than EM.

\begin{figure}
    \centering
    \includegraphics[width=0.9\textwidth]{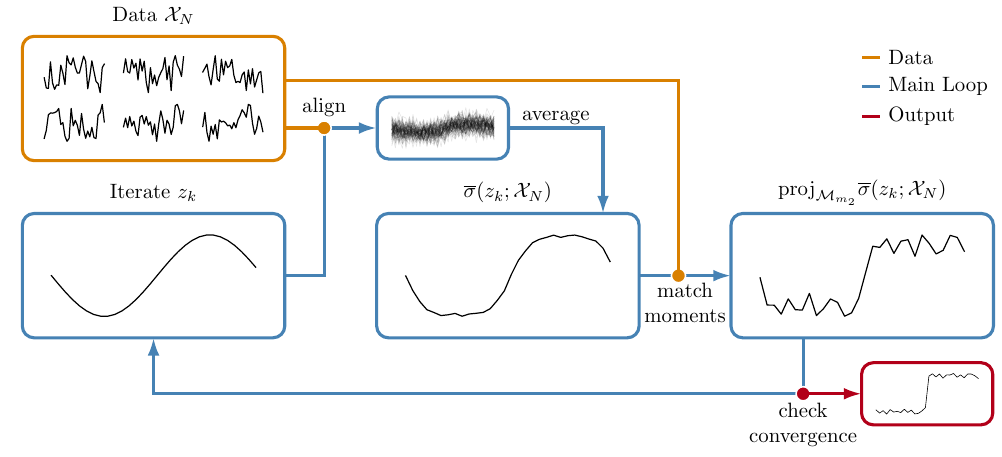}
    \caption{Illustration of the MCA algorithm presented in this paper, applied to a step function signal with $N = 1000$ at $\text{SNR} = 1$. %\textcolor{blue}
    {The arrows indicate data flow between stages of the algorithm -- merging of two edges means that the data is combined. The ``align'' step corresponds to computing $a_i = \falign{z_k}{\xi_i}$ for all $i=1,\dots,N$, the ``average'' step consists of computing $\faligndata{z_k}{\mathcal{X}_N}=\tfrac{1}{N}\sum_{i=1}^N a_i$ and the ``match moments'' step consists of projecting $\faligndata{z_k}{\mathcal{X}_N}$ onto the phase manifold $\mathcal{M}_{m_2}$. The converged output is indicated in red.}}
    \label{fig: algorithm}
    \vspace{-2em}
\end{figure}

In this work, we improve template alignment by combining it with a moment constraint on the reconstruction. This constraint, referred to as the \emph{phase manifold}, requires all signals to have the same power spectrum as $x$, that is, the same second-order moments. %\textcolor{blue}
{We propose an iterative algorithm which we call moment constrained alignment (MCA), that alternates between template alignment and projection onto the phase manifold. The intuition behind this approach is that the alignment takes care of the phases of the Fourier coefficients in the signal, and the projection adjusts the magnitudes. We show that MCA corresponds to a projected block-coordinate descent algorithm with respect to the joint log-likelihood function constrained to the phase manifold, and is therefore provably convergent.}
%We argue that applying template alignment to samples $\xi_i$ in order to align them with a signal $z$ on the phase manifold effectively combines the Fourier phases of both the template and the true signal. Using this phenomenon, we propose a two-step iterative method as follows: we initialize our template with a random guess $z_0$ on the phase manifold of $x$, then apply template alignment, and subsequently project the result back onto the phase manifold. By repeating this process, we approximate the true signal $x$ (see \cref{fig: algorithm}).
We show on a set of test cases that MCA performs comparably to both the method of moments and EM. In particular, we observe that MCA is more robust to noise than the method of moments and also significantly faster than EM.  

\begin{figure}
    \centering
\includegraphics[width=0.95\linewidth]{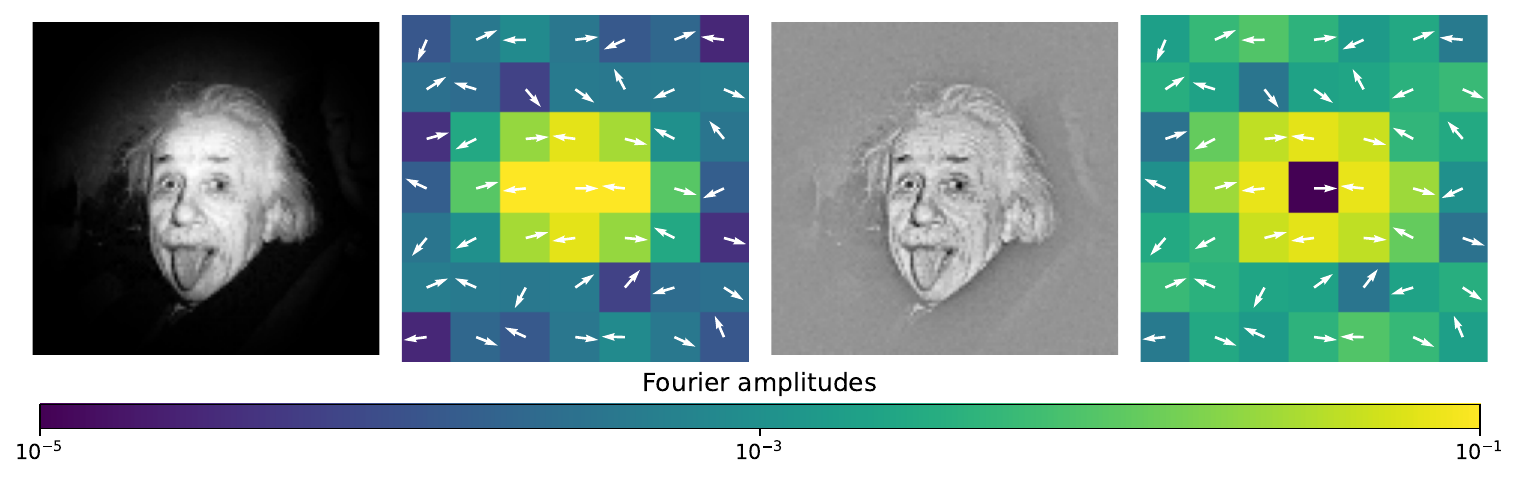}
    \caption{Black and white images: Einstein's iconic photograph (left) and alignment of 80,000 images of pure white noise with that photograph using cross-correlation. Color images: Downsampled Fourier amplitudes of their left-hand side images, with phases indicated as white vectors. The aligned picture was first observed in \cite{shatsky2009method}, %\textcolor{blue}
    {and is related to the phenomenon of confirmation bias from noise, as discussed in \cite{balanov2024confirmation}.}
    }
    \label{fig:albert} 
    \vspace{-2em}
\end{figure}

It is worth mentioning that in an extreme scenario, such as when the SNR is zero -- meaning the observed signals are essentially white noise -- the output Fourier phases of template alignment align with those of the template itself, as illustrated in \cref{fig:albert}. 

The structure of our paper is as follows. In \Cref{sec:prim}, we introduce the basic notions of our approach including the empirical loss and aligning functions, which converge uniformly to the expected loss and aligning functions on the phase manifold. \Cref{sec:phase-recovery} is dedicated to examining the behavior of the expected loss function, including its smoothness and critical points across all SNR regimes. %\textcolor{blue}
{The key results in this section are \cref{prop:loss-smooth}, which proves that the expected loss function is smooth, \cref{prop:gradi-formula}, which derives an explicit expression for the gradient of the loss in terms of the aligning function, and \cref{prop:crit_paral}, which shows that the true signal $x$ is a critical point of the loss function. In \Cref{sec:method}, we state the MCA algorithm, and prove that it results in a sequence of iterates that converge to a stationary loss value}. Finally, in the \Cref{sec:experiments}, we compare the accuracy and performance of our method with those of the EM and higher-order method of moments algorithms.

\section{Preliminary Results}
\label{sec:prim}
%Most MRA methods based on the uniform shift assumption, both aligning and non-aligning, are invariant to circular shifts in the data, and our method also shares this property. We can therefore assume, without loss of \linelabel{line: wlog}generality, that all the shifts are zero. That is, $r_i = 0$ for all $i = 1, 2, \ldots, N$ and therefore the signals are described by
Most MRA methods that assume uniform shifts -- whether aligning or non-aligning -- are invariant to circular shifts in the data, and our method shares this property. Therefore, without loss of generality, we may assume that all shifts are zero; that is, $r_i = 0$ for all $i = 1, 2, \ldots, N$, so the observations reduce to
\begin{equation}
\label{eq:observed_signals_noshift}
    \xi_i = x+\epsilon_i.
\end{equation}
%\textcolor{blue}{This model simplifies the presentation and avoids dense notation.}

%In the MRA setting, one common approach for recovering the true signal is using functions that are invariant under circular translations. One set of such invariant functions are the moments of the signal.
%\textcolor{blue}{A common strategy in the MRA setting is to exploit this invariance under circular translations by using shift-invariant functions. Among the most widely used are the signal’s moments.}
%\textcolor{blue}{A common reason to use shift-invariance in MRA is to avoid explicit shift estimation, which is sensitive to noise (the number of unknown shifts grows as the size of the data set). The signal mean and autocorrelation are two shift-invariant statistics with widespread use in MRA -- they are the first- and second-order moments of the signal, respectively.}
%The zero shift assumption avoids unnecessarily dense notation. 
%\textcolor{blue}{A key advantage of exploiting shift-invariance in MRA is that it removes the need to explicitly estimate the shifts -- an estimation task that is particularly sensitive to noise and grows in complexity with the dataset size.} 

%\textcolor{blue}
{Two statistics widely used in MRA that are invariant to circular shifts are the signal mean and autocorrelation, corresponding to the first- and second-order moments of the signal, respectively. For a real signal $x \in \mathbb{R}^L$, these moments are defined as}

%For a real signal $x \in \mathbb{R}^L$, the first two moments of $x$ are denoted by $m_1(x)\in \mathbb{R}$ and $m_2(x)\in \mathbb{R}^L$, respectively, and are computed as follows:
\begin{equation}
    m_1(x) := \frac{1}{L} \sum_{n=0}^{L-1} x[n]  \qquad
    m_2(x)[n_1] := \frac{1}{L} \sum_{n_2=0}^{L-1} x[n_2]x[n_2 - n_1].
\label{eq:m2}
\end{equation}

%We introduce the notation $\mathcal{X}_N = \{\xi_1,\dots,\xi_N\}$ for the set of all measurements.
For a given sample set $\mathcal{X}_N$, the mean of the signal $x$ can be estimated by  
\[
m_1^*(\mathcal{X}_N)= \frac{1}{N} \sum \limits_{j=1}^N \sum \limits_{n=0}^{L-1} \xi_j[n] \sim \mathcal{N}\left(m_1(x),\frac{\tau^2}{LN}\right).
\]

To estimate the second-order moment $m_2(x)$, let us instead consider its Fourier transform.
Given a signal $x \in \mathbb{R}^L$, its discrete Fourier transform (DFT) $\widehat{x}\in \mathbb{C}^L$ is defined by
\[
    \hat x[m] = \sum_{n=0}^{L-1} x[n]\exp({\frac{2\pi i}{L}nm}).
\]
The second-order moment $m_2(x)$ can be written as a convolution $x\ast \flip x$ between $x$ and a flipped version of itself, $(\flip x)[n] = x[-n]$. Therefore, by the convolution theorem, eq. \cref{eq:m2} translates to $\widehat{m}_2(x)[k] = |\widehat{x}[k]|^2$ in the Fourier domain.
Consequently, the Fourier transform of the second-order moment, known as the \emph{power spectrum} of the signal, provides the information about its Fourier amplitudes. Based on this Fourier formulation, an unbiased estimator $\widehat{m}^*_2$ of ${\widehat{m}}_2(x)$ can be computed as 
\begin{equation}
\widehat{m}^*_2(\mathcal{X}_N)[k] = \frac{1}{N} \sum_{j=1}^N \left(|\widehat{\xi}_j[k]|^2 - L\tau^2 \right ),\label{eq:m2 estimate}
\end{equation}
where the term $L\tau^2$ in \cref{eq:m2 estimate} accounts for the expected contribution of noise, 
$\mathbb{E}[|\widehat{\epsilon}[k]|^2]$, in the power spectrum, ensuring the estimator is unbiased.

For large $\tau$, it is known that the variance of $\widehat{m}^*_2$ is of order $\tau ^4 / N$. In~\cite{bendory2017bispectrum}, it is shown that by considering the third-order moment $m_3(x)$, with Fourier transform $\widehat{m}_3(x)[n_1,n_2]=\widehat{x}[n_1]\overline{\widehat{x}[n_2]}\widehat{x}[n_2-n_1]$, one can reconstruct the true signal up to circular translations. The third-order moments $m_3$ are cubic polynomials and therefore have variance of the order $\tau^6/N$.   %A drawback of using $m_3(x)$ is that it is hard to estimate in the presence of noise. In fact, an estimator based on sample averages (similar to $m^*_1$ and $\widehat{m}^*_2$ above) is a cubic polynomial of noisy signals, and thus has variance of order $\tau^6 / N$.

In our method, we do not incorporate $m_3$. Instead, we restrict the search space to the \emph{phase manifold}, defined by the non-negative vector $m_2 \in \mathbb{R}^L$. To be precise, the phase manifold $\mathcal{M}_{m_2} \subset \mathbb{R}^L$ is defined as follows:
\begin{equation}
\mathcal{M}_{m_2} = \left\{z \in \mathbb{R}^L \ \Bigg| \ m_1(z)=0, \quad 
m_2(z)[i] = m_2[i] \ \text{for} \ i \in \{0, \ldots, L-1\} \right\}.
\end{equation}
In the Fourier domain, this phase manifold is represented by $\widehat{\mathcal{M}}_{{m}_2} \subset \mathbb{C}^L$. 
\begin{remark}
\label{rmk:alpha-torus}
     For odd $L$, a non-empty phase manifold $\mathcal{M}_{m_2}$ is a $d$-dimensional torus centered at the origin, where $d$ is precisely $\#\supp(\widehat{m}_2)$, and $\supp(\widehat{m}_2) := \{i\mid \widehat{m}_2[i]\neq 0,\quad 1\leq i\leq \lfloor\tfrac{L-1}{2}\rfloor\}$. When $L$ is even, the phase manifold is the disjoint union of two antipodal $d$-dimensional tori.
\end{remark}

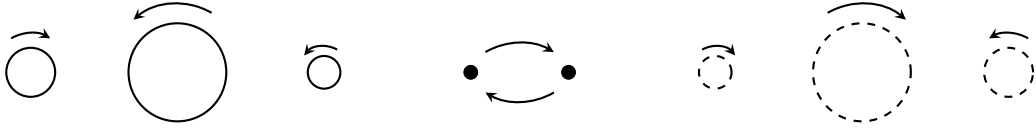
\begin{figure}[t]
    \centering
    \begin{tikzpicture}[scale=1.3] 
        \def\rone{0.25} 
        \def\rtwo{0.5}   
        \def\rthree{0.1667} 
        \def\rfour{0.1667} 
        \def\rfive{0.5}  
        \def\rsix{0.25}

        \draw[thick] (0, 0) circle (\rone);
        \draw[-stealth, thick] (0, 0) ++ (120:\rone+0.15) arc[start angle=120, end angle=60, radius=\rone+0.15];

        \draw[thick] (1.5, 0) circle (\rtwo);
        \draw[-stealth, thick] (1.5,0) ++ (60:\rtwo+0.2) arc[start angle=60, end angle=130, radius=\rtwo+0.2];  

        \draw[thick] (3, 0) circle (\rthree);
        \draw[-stealth, thick] (3,0) ++(60:\rthree+0.1) arc[start angle=60, end angle=140, radius=\rthree+0.1];  

        \filldraw (4.5, 0) circle (2pt);
        \filldraw (5.5, 0) circle (2pt);
        \draw[-stealth, thick] (5, -0.4) ++ (120:0.7) arc[start angle=120, end angle=60, radius=0.7];
        \draw[-stealth, thick] (5, 0.4) ++ (-60:0.7) arc[start angle=-60, end angle=-120, radius=0.7];

        \draw[dashed, thick] (7, 0) circle (\rfour);
        \draw[ -stealth, thick] (7, 0) ++ (120:\rfour+0.1) arc[start angle=120, end angle=40, radius=\rfour+0.1];

        \draw[dashed, thick] (8.5, 0) circle (\rfive);
        \draw[-stealth, thick] (8.5,0) ++(120:\rfive+0.2) arc[start angle=120, end angle=50, radius=\rfive+0.2];  

        \draw[dashed, thick] (10, 0) circle (\rsix);
        \draw[-stealth, thick] (10,0) ++ (60:\rsix+0.15) arc[start angle=60, end angle=120, radius=\rsix+0.15];  
    \end{tikzpicture}
    \caption{A group element in \( \operatorname{CO}(8) \) acting on a 3-dimensional phase manifold. The solid circles on the left rotate under this group element (the dashed circles represent the negative frequencies with conjugate phase). The permutation of the pair of points in the middle is due to the eigenvalue \( -1 \) of the group element.}
    \label{fig:circ_CO}
    \vspace{-2em}
\end{figure}

We now introduce two important group actions on phase manifolds that act as isometries.
\begin{definition}
    Let $\operatorname{Circ}(L)$, $\operatorname{O}(L)$, and $\operatorname{SO}(L)$ be the space of circulant, orthogonal, and special orthogonal matrices of size $L\times L$, respectively. We define the groups $\co(L)$ and $\cso(L)$ to be $\operatorname{Circ}(L) \cap \operatorname{O}(L)$ and $\operatorname{Circ}(L) \cap \operatorname{SO}(L)$, respectively.
\end{definition}
\begin{remark}
\label{rmk:transit}
When $L$ is odd, the group action $\cso(L)$ acts transitively on $\mathcal{M}_{m_2}$. (When $L$ is even, $\co(L)$ acts transitively on $\mathcal{M}_{m_2}$.) In the Fourier domain, these group actions act on the $k$th Fourier mode for $1\le k\le(L-1)/2$ via the unitary group $\text{U}(1)$ (see \cref{fig:circ_CO}).
Also, if $x$ and $z$ are two signals in a phase manifold $\mathcal{M}_{m_2}$, with $x$ having zero Fourier phase, then there exists a matrix $C_z \in \co(L)$ such that $z=C_z x$. Furthermore, $e^{i\theta_k}$ is an eigenvalue of $C_z$ whenever the $k$th Fourier mode $\widehat{z}[k] = |\widehat{z}[k]| e^{i\theta_k}$ is non-zero. 
\end{remark}
\subsection{Discriminant and Alignment}
\label{subsec:Align}
\begin{figure}[ht]
    \definecolor{left}{rgb}{0.85, 0.50, 0.0}
    \definecolor{fade}{rgb}{1.0, 0.9, 0.0}
    \centering
    \begin{tikzpicture}
        \draw (0,0) node {\includegraphics[width=0.4\textwidth]{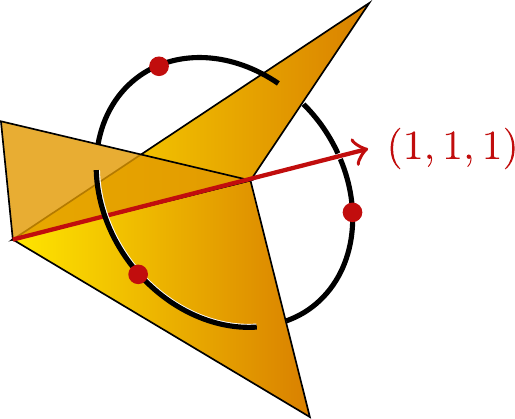}};
        \draw (1.3,-0.2) node[right] {$\pmb{\sigma z}$};
        \draw (-1.4,2.1) node[left] {$\pmb{\sigma_2 z}$};
        \draw (-1.7,-0.9) node[left] {$\pmb{z}$};
    \end{tikzpicture}
    \caption{Three intersecting half-planes, depicted in gold with $120$° dihedral angles, symbolize the total discriminant of any signal $z$ in $\mathbb{R}^3 \setminus \langle \mathbbm{1}_3 \rangle_{\mathbb{R}}$. The phase manifold is a circle, intersecting the total discriminant at three distinct points.}
    \label{fig:discrim1}
\end{figure}
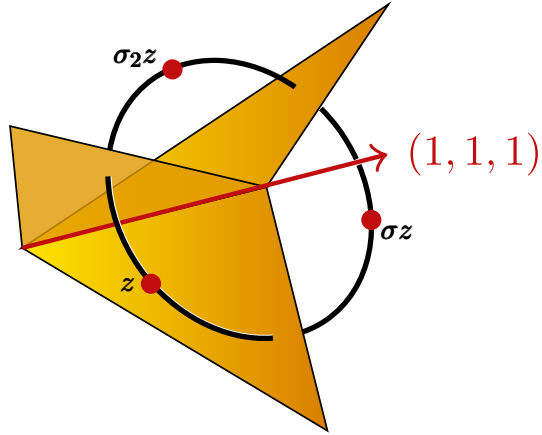

We define the \emph{aligning function} \(\sigma: (\mathbb{R}^L \setminus \mathcal{P}_L) \times \mathbb{R}^L\to \mathbb{R}^L\) such that \(\falign{z}{x} = \sigma_i(x)\), where \(i = \argmin_{k} ||z - \sigma_k(x)||\) and $\mathcal{P}_L$ denotes the space of all periodic signals in $\mathbb{R}^L$. When $L$ is a prime integer, then $\mathcal{P}_L$ is the vector space spanned by the vector $\mathbbm{1}_L=(1,\ldots,1)^\top \in \mathbb{R}^L$.
Moreover, for a signal $x$ in a general position, there exists an open set \(U_x\) containing \(x\) where the minimizing index \(i\) with respect to $z$ remains constant for every \(z \in U_x\). In the following definition, we characterize the set of all vectors $t\in \mathbb{R}^L$, where the mapping \(t \mapsto \argmin_{k} ||z - \sigma_k(t)||\) fails to be locally constant for a fixed $z$. In other words, $t$ is simultaneously closest to at least two signals in the set $\{z,\sigma_1(z),\ldots, \sigma_{L-1}(z)\}$.

\begin{definition}
\label{def:disc}
    Let $z$ be a vector in $\mathbb{R}^L \setminus \mathcal{P}_L$. We define the discriminant $\Delta^{i,j}_z$ to be a polyhedral set comprising vectors $t \in \mathbb{R}^L$ such that $\|t - \sigma_i(z)\| = \|t - \sigma_j(z)\| \le \|t - \sigma_k(z)\|$ for every $k$. Additionally, we define the \emph{total discriminant} $\Delta_z$ as the union $\bigcup_{i,j} \Delta_z^{i,j}$; refer to \cref{fig:discrim1,fig:discrim2}.
\end{definition}

We use this discriminant to characterize the behavior of the aligning function $\falign{z}{x}$.

\begin{remark}
\label{rmk:partitions_delta}
    For any vector $z \in \mathbb{R}^L \setminus \mathcal{P}_L$, the set $\mathbb{R}^L \setminus \Delta_z$ is the union of $L$ open polyhedral cones $C_i(z)$ such that $\sigma_j(C_i(z))=C_{i+j}(z)$, and $C_i(z)$ contains $\sigma_i(z)$. Moreover, the function $x\mapsto \falign{z}{x}$ is equal to $\sigma_{-i}(x)$ when $x \in C_i(z)$ and $z\mapsto \falign{z}{x}$ is a locally constant function when $x \in C_i(z)$.
\end{remark}
One simple observation reveals that the discriminant $\Delta_z^{i,j}$ that satisfies the conditions in \cref{def:disc} has codimension one and
 the vectors $\lambda \cdot \mathbbm{1}_L$ is always contained in each $\Delta_z^{i,j}$ and thus the total discriminant $\Delta_z$. Also, $\Delta_z^{i,j}$ has a normal vector $\mathbf{n}_{i,j} = \sigma_i(z)-\sigma_j(z)$, hence we can compute the dihedral angle $\varphi$ between discriminants $\Delta_z^{i,j}$ and $\Delta_z^{i',j'}$ as follows:
 \begin{equation}
     \label{eq:dih}
     \cos \varphi = \frac{|\mathbf{n}_{i,j} \cdot \mathbf{n}_{i',j'}|}{||\mathbf{n}_{i,j}||||\mathbf{n}_{i',j'}||}.
 \end{equation}
To compute \(\mathbf{n}_{i,j} \cdot \mathbf{n}_{i',j'}\) in \cref{eq:dih}, it suffices to know \(\langle\sigma_{k}(z), \sigma_{\ell}(z)\rangle\) for all pairs $k, \ell \in \{0, 1, \ldots, L-1\}$. Consequently, the computation of dihedral angles \(\varphi\) is encoded in the second-order moment \(m_2(z)\). Thus, for all \(z\) on the phase manifold, the relative positions of discriminants remain unchanged. More precisely, suppose that \(C \in \co(L)\), then for any \(i, j\), one can see that \(\Delta^{i,j}_{Cz} = C \Delta^{i,j}_z\). Consequently, we have that \(\Delta_{Cz} = C \Delta_z\).

Although the aligning function is technically undefined on the discriminant, the probability of any signal being on the discriminant is zero under the MRA model~\cref{eq:observ-sig}. Therefore, as long as the template is generic, we can in practice compute the aligning function on all data without ambiguity:
%Let us now consider aligning all the signals in our dataset $\mathcal{X}_N$.

\begin{definition}
\label{def:align-func}
Let $z \in \mathbb{R}^L \setminus \mathcal{P}_L$ and let $\mathcal{X}_N=\{\xi_1,\ldots,\xi_N\}$ be a set of observed signals as described in \cref{eq:observed_signals_noshift}. We define the \emph{averaged aligning function} $\overline{\sigma}(z, \bullet):\mathbb{R}^{L \times N} \to \mathbb{R}^L$ as:
\begin{equation}
\label{eq:align-func}
\faligndata{z}{\mathcal{X}_N} =\frac{1}{N}\sum_{i=1}^N  \falign{z}{ \xi_i}.
\end{equation}
\end{definition}

We also need a way to measure to what extent our estimate $z$ aligns with the data.
\begin{definition}
\label{def:loss_func}
    Let $z$ be a non-periodic signal on the phase manifold $\mathcal{M}_{m_2}$. For a given set of observed signals $\mathcal{X}_N = \{\xi_1,\ldots,\xi_N\}$ as described in \cref{eq:observed_signals_noshift}, we define the \emph{expected loss function} $\lossdata{\bullet}{\mathcal{X}_N}: \mathcal{M}_{m_2} \to \mathbb{R}$ to be (see \cref{fig:loss-L=5})
    \begin{equation}
    \lossdata{z}{\mathcal{X}_N} =\frac{1}{2N} \sum \limits_{i=1}^N ||z-\falign{z}{\xi_i}||^2.
        \end{equation} 
\end{definition}
{We remark that the averaged aligning function $\overline\sigma$ is closely related to the gradient of the loss function. Specifically, $\nabla_z\lossdata{z}{\mathcal{X}_N}=z-\faligndata{z}{\mathcal{X}_N}$ when $z$ is not on the discriminants $\cup_{i=1}^N\Delta_{\xi_i}$. We will soon see that the same identity also holds in the limit of infinite data.} 

An important regime for studying the behavior of the averaged aligning and average loss function is that of the infinite-data regime. In this case, we let $N \to \infty$ and note that the terms in each sum are bounded in expectation since $\mathbb{E}[\falign{z}{\xi_i}[n]]\leq \mathbb{E}[\max_{\ell}\xi_i[\ell]]$ and $\mathbb{E}[\|z-\falign{z}{\xi_i}\|]\leq \|z\|+\mathbb{E}[\|\xi_i\|]$. By the law of large numbers, both sums therefore converge in probability to their expected values, which only depend on $z$, $x$ (the clean signal) and $\tau$ (the standard deviation of the noise).

\begin{definition}
Let $z \in \mathbb{R}^L \setminus \mathcal{P}_L$, $x \in \mathbb{R}^L$ and $\tau \ge 0$, then we define the infinite-data versions of the averaged aligning function
\begin{equation}
    \falignlim{z}{x} := \lim \limits_{N\to \infty} \faligndata{z}{\mathcal{X}_N}
\end{equation}
and the expected loss function
\begin{equation}
   \losslim{z}{x} = \lim \limits_{N\to \infty} \lossdata{z}{\mathcal{X}_N}
\end{equation}
\end{definition}

We proceed to study the properties of these functions. First, we note that each can be written out explicitly as an integral over the possible noise realizations.
\begin{remark}
    Let $z \in \mathbb{R}^L\setminus \mathcal{P}_L$, $x \in \mathbb{R}^L$ and $\tau > 0$, and define the Gaussian density function $\phi_\tau(\epsilon)=(2\pi\tau^2)^{-L/2}\exp(-\|\epsilon\|^2/(2\tau^2))$. Then
    the averaged aligning function $\falignlim{z}{x}$ and the expected loss function $\losslim{z}{x}$ can be written as follows:
    \begin{align}
    \label{eq:averaged-aligned-integral}
    \falignlim{z}{x} &= \mathbb{E}[\falign{z}{x+\epsilon}]= \int_{\mathbb{R}^L} \falign{z}{x+y}\phi_\tau(y)\; \mathrm{d}\mu(y),\\
    \label{eq:loss-integral}
        \losslim{z}{x} &=\mathbb{E}[\tfrac{1}{2}\|z-\falign{z}{x+\epsilon}\|^2] = \int_{\mathbb{R}^L} \frac{1}{2}||z-\falign{z}{x+y}||^2\phi_\tau(y)\; \mathrm{d}\mu(y).
    \end{align}
\end{remark}
Second, we see that both functions exhibit some useful scaling properties.
\begin{lemma}
\label{lem:properties_L_sigma}
Let $z$ be a vector in $\mathbb{R}^L \setminus \mathcal{P}_L$ and $x \in \mathcal{M}_{m_2}$. For all constants $\lambda,\tau>0$,
    \begin{enumerate}[label=(\roman*), leftmargin=*, align=left]
    \item $\losslim{z}{x} = \tau^2\cdot \losslimgen{1}{\tfrac{z}{\tau}}{\tfrac{x}{\tau}}$,
    \item $\falignlim{z}{x} = \tau \cdot \falignlimgen{1}{\frac{z}{\tau}}{\frac{x}{\tau}}$,
    \item $m_1(x) = m_1(\falign{z}{x})= m_1(\falignlimgen{0}{z}{x})=m_1(\falignlim{z}{x})$,
    \item $\falignlim{\lambda \cdot z}{x}=\falignlim{z}{x} $.
\end{enumerate}
\end{lemma}
\begin{proof}
    The first two identities are obtained by change of variable $\frac{y}{\tau} \to y$ in \cref{eq:loss-integral} and \cref{eq:averaged-aligned-integral}. The third identity follows from the fact that alignment does not change the mean and that in our setting, the i.i.d. random normal vector $\epsilon$ has zero mean.
    For the last identity, note that for every positive real number $\lambda$, we have \[\underset{k}{\argmax}\left(\sum \limits_{m=0}^{L-1} z[m]x[m+k]\right)=\underset{k}{\argmax}\left(\sum \limits_{m=0}^{L-1} \lambda\cdot z[m]x[m+k]\right),\] thus $\falignlim{\lambda \cdot z}{x}=\falignlim{z}{x}$.
\end{proof}

\begin{figure}[t]
    \centering
    \includegraphics[width=\linewidth]{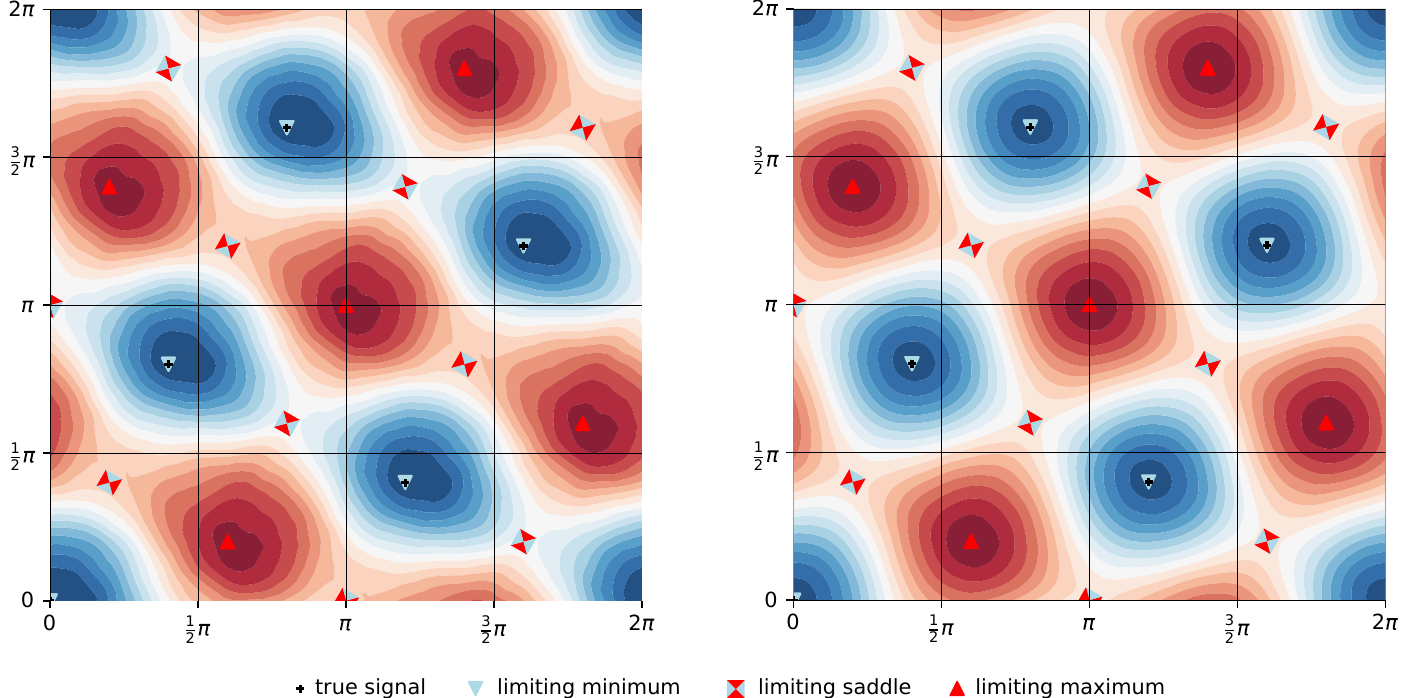}
    \caption{Loss function $\lossdata{z}{\mathcal{X}_N}$ with data generated from the true signal $x=(4/5, -1/5, -1/5, -1/5, -1/5)^\top$ and number of samples $N=10^5$, as a function of the phases $\hat z[1]$ and $\hat z[2]$. The left figure uses $\tau = 1.7$, and the right figure uses $\tau = 1$. Due to \cref{rmk:alpha-torus}, the phase manifold $\mathcal{M}_{m_2}$ is isometric to a $2$-torus, represented here as a rectangle. The critical points of the infinite-data loss function $\losslim{z}{x}$ are marked.}
    \label{fig:loss-L=5}
    \vspace{-2em}
\end{figure}
Despite the straightforward construction of the averaged aligning function $\falignlim{z}{x}$, there is no known closed formula describing $\falignlim{z}{x}$ as a function of  $z$, $x$, and $\tau$. In the following example, we see that even in a special case, obtaining a concise description of this function remains challenging.
\begin{example}
\label{ex:align_x=0}
    Let \( z = (1, 0, \ldots, 0)^\top \) and \( \tau > 0 \). Using \cref{lem:properties_L_sigma}, we have
    \begin{equation}
        \label{eq:property_delta}
        \falignlim{z}{\mathbf{0}_L} = \tau \cdot \falignlimgen{1}{\tfrac{z}{\tau}}{\mathbf{0}_L} \quad \text{and} \quad
    m_1(\falignlim{z}{\mathbf{0}_L}) = 0.
    \end{equation}
    In this case, it is possible to show that \(\falignlimgen{1}{z}{\mathbf{0}_L}[i] = \falignlimgen{1}{z}{\mathbf{0}_L}[j] \) for every \( i, j \neq 0 \).
     Hence,
    \[
    \falignlim{z}{\mathbf{0}_L} = \tau \cdot \left( a_L, \frac{-a_L}{L-1}, \ldots, \frac{-a_L}{L-1} \right)^\top,
    \]
    thus similar to $z$, the vector \( \falignlim{z}{\mathbf{0}_L} \) is circularly symmetric and therefore has zero Fourier phases.
    
    Now, to compute \( a_L \), it is enough to assume \( \tau = 1 \) and find the value of \( \falignlimgen{1}{z}{\mathbf{0}_L}[0] \). This value can be interpreted as finding \( \mathbb{E}[Z_L] \), where \( Z_L = \max(X_1, \ldots, X_L) \) and \( X_i \) are i.i.d. standard normal variables. Let \( \Phi_1 \) denote the cumulative distribution function of the standard normal distribution with density function $\phi_1$. Then,
    \[
    a_L = \mathbb{E}[Z_L] = \int_{-\infty}^{\infty} \mathrm{x} \left( \Phi_1(\mathrm{x}) \right)^{L-1} \phi_1(\mathrm{x}) \, d\mathrm{x}.
    \]
    This has no closed-form solution, but for \( L \in \{2, \ldots, 9\} \), we can compute \( a_L \) numerically:
    \begin{align*}
        a_2 & \approx 0.56418, & a_3 & \approx 0.84628, & a_4 & \approx 1.02937, & a_5 & \approx 1.16296, \\
        a_6 & \approx 1.26720, & a_7 & \approx 1.35217, & a_8 & \approx 1.42360, & a_9 & \approx 1.48501.
    \end{align*}
    \end{example}

\subsection{Convergence to Infinite-Data Limit}
We turn now to the relation between the finite data loss $\lossdata{\bullet}{\mathcal{X}_N}$, infinite-data loss function $\losslim{\bullet}{x}$ and aligning function $\falignlim{\bullet}{x}$. We refer to~\cref{apx: lipschitz} for the proofs in this section. We first establish that $\lossdata{\bullet}{\mathcal{X}_N}$ is Lipschitz. In fact, the distance function $z\mapsto \|z-\falign{x}{z}\|^2$ itself is Lipschitz on any compact set, as the following lemma shows:
\begin{lemma}\label{lem: lipschitz}
    For any elements $u,v$ and $w$ in $\mathbb{R}^L$, it holds that
    \[
        \tfrac{1}{2}\|u - \falign{w}{u}\|^2 - \tfrac{1}{2}\|v - \falign{w}{v}\|^2 \leq \|u-v\|(\tfrac{1}{2}\|u+v\|+\|w\|).
    \]
\end{lemma}

We can now use this to prove the Lipschitz continuity of the loss function.

\begin{proposition}\label{prop: lipschitz}
    For any $u,v\in \mathcal{M}_{m_2}$, the following holds:
    \[
        |\lossdata{u}{\mathcal{X}_N}-\lossdata{v}{\mathcal{X}_N}|\leq C_N \|u-v\|\quad\text{with}\quad C_N = 2\|m_2\| + \frac{1}{N}\sum_{n=1}^N\|\epsilon_n\|.
    \]
    %Furthermore, $C_N\overset{ \small a.s.}{\rightarrow} 2m+\sqrt{2}\tau\frac{\Gamma((L+1)/2)}{\Gamma(L/2)}$ as $N\to\infty$.
    Furthermore, $\mathbb{P}(|C_N-2\|m_2\|-\sqrt{2}\tau\frac{\Gamma((L+1)/2)}{\Gamma(L/2)}|\geq \epsilon)\leq \frac{L\tau^2}{\epsilon^2N}$.
\end{proposition}

\Cref{prop: lipschitz} lets us control the variability of the Lipschitz constant of the empirical loss function. This is enough to guarantee uniform convergence to the expected loss when restricted to the phase manifold at a rate of $\mathcal{O}(1/N)$, as the following proposition shows:

\begin{proposition}\label{prop: c0 convergence}
    The loss function converges uniformly in probability to its infinite data limit on the phase manifold $\mathcal{M}_{m_2}$ with rate $N^{-1}$. That is, for an $\epsilon>0$ sufficiently small:
    \[
        \mathbb{P}\left\{\sup_{z\in\mathcal{M}_{m_2}}|\lossdata{z}{\mathcal{X}_N}-\losslim{z}{x}|\geq \epsilon\right\}\leq \frac{4(\tau^2L+\|\hat m_2\|_\infty^2)\|\hat m_2\|_\infty^L}{N\epsilon^{L+2}},
    \]
    where $\|\hat m_2\|_\infty=\max_k |\hat m_2[k]|$ is the maximal Fourier magnitude of the signal $x$.
    %\[
    %    \mathbb{P}\left\{\lim_{N\to\infty}\sup_{z\in\mathcal{M}}|\lossdata{z}{\mathcal{X}_N} - \losslim{z}{x}|=0\right\} = 1.
    %\]
\end{proposition}

We lastly comment on the meaning of \cref{prop: c0 convergence}. Due to convergence in the supremum norm, the set of minimizers in the discrete data case must converge to the minimizers of the infinite data function. 
\begin{corollary}\label{cor: convergence of minimizers}
    Let $z_N$ minimize the discrete loss $\lossdata{\bullet}{\mathcal{X}_N}$. Then, with probability one, $\lim_{N\to\infty}\losslim{z_N}{x}=\min_{z\in\mathcal{M}}\losslim{z}{x}$, that is $z_N$ attains optimal loss as $N\to\infty$. Furthermore, with probability one, $z_N$ converges to the set $Z^*:=\arg\min_{z\in\mathcal{M}}\losslim{\bullet}{x}$ of minimizers of the infinite data loss, with respect to the minimum distance metric.
\end{corollary}

\begin{proof}
    Let $(z_{N_k})_{k=1}^\infty$ denote an arbitrary convergent subsequence of $(z_N)_{N=1}^\infty$, and denote its limit point by $z_\infty$. Take any $z^*\in Z^*$. By \cref{prop: c0 convergence}, we have with probability one, that
    \begin{align*}
        \losslim{z_\infty}{x} - \min_{z\in\mathcal{M}}\losslim{z}{x} &=\lim_{k\to\infty}[\losslim{z_\infty}{x}-\lossdata{z_\infty}{\mathcal{X}_{N_k}}]+[\lossdata{z_\infty}{\mathcal{X}_{N_k}}-\lossdata{z_{N_k}}{\mathcal{X}_{N_k}}]\\
        &+[\lossdata{z_{N_k}}{\mathcal{X}_{N_k}}-\lossdata{z^*}{\mathcal{X}_{N_k}}]+[\lossdata{z^*}{\mathcal{X}_{N_k}}-\losslim{z^*}{x}]\\
        &=\lim_{k\to\infty}\lossdata{z_{N_k}}{\mathcal{X}_{N_k}}-\lossdata{z^*}{\mathcal{X}_{N_k}} \leq 0,
    \end{align*}
    where we used optimality of $z_{N_k}$ for the data $\mathcal{X}_{N_k}$ in the last step. Hence, any convergent subsequence of $(z_N)_{N=1}^\infty$ converges to a minimizer of the infinite data loss function $\losslim{\bullet}{x}$. By continuity of $\losslim{\bullet}{x}$, the sequence $(\losslim{z_{N}}{x})_{N=1}^\infty$ therefore converges to $\min_{z\in\mathcal{M}}\losslim{z}{x}$. Lastly, suppose the full sequence $(z_N)_N$ does not converge to $Z^*$ in the minimum distance. Then, we can form an infinite subsequence $(z_{N_\ell})_{\ell=1}^\infty$ so that $|z_{N_\ell}-z|>\epsilon$ for all $z\in Z^*$. But then, there is a convergent subsequence $(z_{N_{\ell_k}})_{k=1}^\infty$ of that sequence, that does not converge to $Z^*$. this is a contradiction, since all convergent subsequences of $(z_N)_{N=1}^\infty$ must converge to a minimizer.
\end{proof}

From \cref{cor: convergence of minimizers} it becomes clear that a key component to understanding the behavior of the minimizers of the finite-data loss function is to study the infinite data limit. In the following sections, we will prove some important properties of the expected loss and how its gradient relates to the aligning function.
\section{Phase Recovery}
\label{sec:phase-recovery}

%%%%%%%%%%%%%%%%%%%%%%%%%%%%%%%%%%%%%%%%%%%%%%%%%%%%%%%%
%%%%%%%%%%%%%%%%%%%%%%%%%%%%%%%%%%%%%%%%%%%%%%%%%%%%%%%%
%%%%%%%%%%%%%%%%%%%%%%%%%%%%%%%%%%%%%%%%%%%%%%%%%%%%%%%%
%%%%%%%%%%%%%%%%%%%%%%%%%%%%%%%%%%%%%%%%%%%%%%%%%%%%%%%%
%%%%%%%%%%%%%%%%%%%%%%%%%%%%%%%%%%%%%%%%%%%%%%%%%%%%%%%%
%%%%%%%%%%%%%%%%%%%%%%%%%%%%%%%%%%%%%%%%%%%%%%%%%%%%%%%%

%%%%%%%%%%%%%%%%%%%%%%%%%%%%%%%%%%%%%%%%%%%%%%%%%%%%%%%%
%%%%%%%%%%%%%%%%%%%%%%%%%%%%%%%%%%%%%%%%%%%%%%%%%%%%%%%%
%%%%%%%%%%%%%%%%%%%%%%%%%%%%%%%%%%%%%%%%%%%%%%%%%%%%%%%%
%%%%%%%%%%%%%%%%%%%%%%%%%%%%%%%%%%%%%%%%%%%%%%%%%%%%%%%%
%%%%%%%%%%%%%%%%%%%%%%%%%%%%%%%%%%%%%%%%%%%%%%%%%%%%%%%%
%%%%%%%%%%%%%%%%%%%%%%%%%%%%%%%%%%%%%%%%%%%%%%%%%%%%%%%%
%\textcolor{blue}
{A key property of any good statistical estimator such as EM, phase synchronization or MCA (ours), is consistency, meaning that the estimates $x_N$ of the true signal $x$ converge to $x$ in the limit $N\to \infty$ of infinite data. In our case, \cref{cor: convergence of minimizers} implies that MCA is consistent if $x$ is the unique global minimizer (up to shift) of $\losslim{\bullet}{x}$. Proving this is a hard problem in the general case. Even showing that $x$ is a local minimum requires good estimates of the second-order derivatives of $\mathcal{L}$ around $x$. The purpose of this section is to show that $x$ is a critical point of $\mathcal{L}$, and in some cases also a local minimum. } 

The section is structured as follows. First, we explore the properties of the expected loss function \(\losslim{z}{x}\) and its relation to the averaged aligning function \(\falignlim{z}{x}\). We prove that \( z \mapsto \losslim{z}{x}\) is a smooth function for \(\tau > 0\) and propose a finite set of critical points for this function -- one out of which is the true signal. Furthermore, we discuss the nature of these critical points, i.e., whether they are minima, saddles, or maxima, and examine the behavior of the expected loss function as \(\tau \to \infty\). %Finally, we conclude this section by using these theoretical results to design a robust algorithm for solving the MRA problem.

\subsection{Properties of the Expected Loss Function}
We begin by proving a useful lemma concerning symmetries of the expected loss and the averaged aligning function. Specifically, we consider the action of the group $ \cso(L) \rtimes \langle \flip \rangle$, where $\langle \flip \rangle \cong \mathbb{Z}_2$ is the group generated by the reflection matrix $\flip$, acting on $\cso(L)$ by conjugation; that is, $\flip C \flip = C^{-1}$ for all $C \in \cso(L)$.

\begin{figure}
    \centering
    \includegraphics[width=1.0\linewidth]{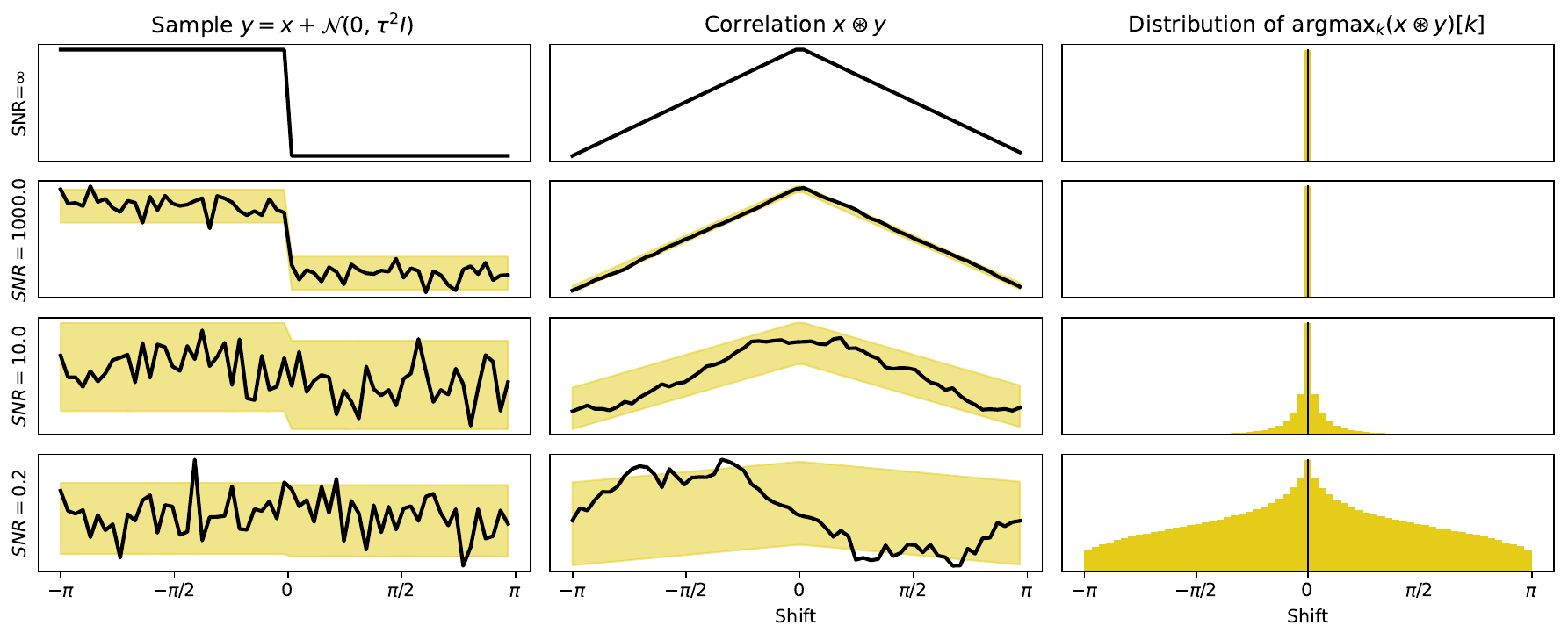}
    \caption{Samples of the unshifted data (left) correlation between the true underlying signal and samples (middle) and distribution of the estimated shifts over $10^7$ samples (right), at different SNR. The $96\%$ confidence intervals are marked in yellow in the left and middle plots.}
    \label{fig: data example}
\end{figure}

\begin{lemma}
    \label{lem:invariant-loss}
    Let $C$ be an element of the group $\cso(L)\rtimes \langle \flip\rangle$. Then, for any two signals $z,x$ of length $L$ on a phase manyfold $\mathcal{M}_{m_2}$, we have
    \[
        \losslim{Cz}{Cx} = \losslim{z}{x},\quad \text{and}\quad \falignlim{Cz}{Cx} = C\falignlim{z}{x}.
    \]
\end{lemma}
\iffalse
    \[
        \losslim{z}{x} = \int_{\mathbb{R}^L}\tfrac{1}{2}\min_\sigma\|x+\epsilon-\sigma z\|^2\phi_\tau(\epsilon)\mathrm{d}\epsilon = \tfrac{1}{2}\|x\|^2+\tfrac{1}{2}\tau^2-\int_{\mathbb{R}^L}\max_\sigma \langle x+\epsilon,\sigma z\rangle \phi_\tau(\epsilon)\mathrm{d}\epsilon.
    \]
\fi
\begin{proof}
    Take $z$ and $x$ in $\mathcal{M}_{m_2}$. Let $S$ be the subgroup of $\mathrm{O}(L)$ containing all discrete cyclic shifts on $\mathbb{R}^L$. Then if $C\in\cso(L)\rtimes \langle \flip\rangle$, it holds that $CSC^{-1} = S$. Furthermore, $\eta = C\epsilon$ is identical in distribution to $\epsilon$, and
    \begin{align*}
        C\falign{z}{x+\epsilon} &= C\left(\argmin_{\sigma\in S}\|z-\sigma(x+\epsilon)\|^2\right)(x+\epsilon)= \\&=C\left(\argmin_{\sigma\in S}\|Cz-(C\sigma C^{-1})(Cx+C\epsilon)\|^2\right)(x+\epsilon)= \\&=\quad \left(\argmin_{\sigma\in CSC^{-1}}\|Cz-\sigma(Cx+C\epsilon)\|^2\right)(Cx+\eta) = \falign{Cz}{Cx+\eta}.
    \end{align*}
    From the above relation, we obtain
    \[
        \losslim{z}{x} = \mathbb{E}\left[\tfrac{1}{2}\|z-\falign{z}{x+\epsilon}\|^2\right] =\mathbb{E}\left[\tfrac{1}{2}\| C z-\falign{Cz}{C x+\eta}\|^2\right] = \losslim{Cz}{Cx}.
    \]
    Similarly:
    \[
        C\falignlim{z}{x} = C\mathbb{E}\left[\falign{z}{x+\epsilon}\right]=C\mathbb{E}\left[C^{-1}\falign{Cz}{Cx+\eta}\right]=\falignlim{Cz}{Cx}.
    \]
\end{proof}

A second useful property is that the expected loss function $z\mapsto\losslim{z}{x}$ can be expressed as a convolution between the square distance to the true signal $x$ and the Gaussian distribution:

\begin{theorem}\label{prop:loss-smooth}
    Let $x\in\mathbb{R}^L$ be a signal on a phase manifold $\mathcal{M}_{m_2}$. Let $\phi_\tau$ denote the $L$-dimensional iid Gaussian density, $\losslimgen{0}{z}{x}$ denote the  square distance $\|x-\falign{z}{x}\|^2$ from the orbit of $x$ to $z$ and let $\circledast$ be the $L$-dimensional convolution (integral). Then, for any $z\in \mathcal{M}_{m_2}$,
    \[
        \losslim{z}{x}    = \left(\phi_\tau\circledast \losslimgen{0}{\bullet}{x}\right)(z).
    \]
    As a consequence, $\losslim{\bullet}{x}$ is smooth over $\mathcal{M}_{m_2}$.
\end{theorem}

\begin{proof}
    It is immediately clear from \cref{eq:loss-integral} that the expected loss is a convolution in the $x$-variable. We now prove using \cref{lem:invariant-loss}, that the variables $x$ and $z$ are interchangeable in the loss function. To see this, consider the following transform $C(x,z)$, defined by its action $\widehat{C}(x,z)$ on the Fourier space:
    \[
        (\widehat{C}(x,z)\widehat{y})[k] = \begin{cases}
            (\widehat{z}[k]\widehat{x}[k]m_2[k]^{-2})\overline{\widehat{y}[k]},&\text{if } m_2[k] \neq 0\\
            \overline{\widehat{y}[k]},&\text{otherwise}.
        \end{cases}
    \]
    Then, it holds that $C(x,z)x = z$ and $C(x,z)z = x$. Furthermore, $C(x,z)$ is a composition of a reflection $\widehat{y}[k] \mapsto \overline{\hat y[k]}$ and a circulant matrix (convolution) $\widehat{y}[k]\mapsto \widehat{x}[k]\widehat{y}[k]m_2[k]^{-2}$ so, by \cref{lem:invariant-loss}, 
    \[
        \losslim{z}{x} = \losslim{C(x,z)z}{C(x,z)x} = \losslim{x}{z}.
    \]
    the result follows simply by switching $x$ and $z$ into~\cref{eq:loss-integral}. Smoothness then follows since differentiation and convolution commute, i.e. $(f\circledast g)'=f'\circledast g$ for functions $f$ and $g$.
\end{proof}

To further understand the behavior of the expected loss function $\losslim{z}{x}$, we first compute its (Euclidean) gradient with respect to $z$, denoted by $\nabla_z \losslim{z}{x}$.
For a finite sample set $\mathcal{X}_N$ and a generic $z \in \mathcal{M}_{m_2}$, we have $\nabla_z \lossdata{z}{\mathcal{X}_N} = z - \faligndata{z}{\mathcal{X}_N}$, since $\faligndata{z}{\mathcal{X}_N}$ is locally constant in $z$ (see \cref{rmk:partitions_delta}). In contrast, with an infinite number of samples, \(\falignlim{z}{x}\) is no longer locally constant in \(z\) for $\tau>0$. Therefore, we cannot directly conclude that \(\nabla_{z} \losslim{z}{x} = z - \falignlim{z}{x}\). In the following proposition, we show that this equality still holds, but a different approach is required.

\begin{proposition}
\label{prop:gradi-formula}
    Let $x$ be a signal on a phase manifold $\mathcal{M}_{m_2}$, and $\tau>0$. Then for all $z\in\mathbb{R}^L$, it holds that $\nabla_{z} \losslim{z}{x} = z - \falignlim{z}{x}$.
\end{proposition}
\begin{proof}
    The proof idea is to apply Lebesgue's dominated convergence theorem with the Gaussian measure $\mathrm{d}P(\epsilon) = \phi_\tau(\epsilon)\mathrm{d}\mu(\epsilon)$ to change the order of integration and differentiation. Given that we are allowed to switch integration and differentiation, the result follows from a simple computation:
    \[
        \nabla_z \losslim{z}{x} = \nabla_z\mathbb{E}\left[\tfrac{1}{2}\|z-\falign{z}{x+\epsilon}\|^2\right] = \mathbb{E}\left[z-\falign {z}{x+\epsilon}\right] = z - \falignlim{z}{x}.
    \]
    Now, to see why this operation is allowed, we fix $z,x$ and let $\nu\in\mathbb{R}^L$ be an arbitrary unit vector. Then, the sequence
    \[
    f_n(\epsilon):=n\left(\tfrac{1}{2}\|z+n^{-1}\nu-\falign{x+\epsilon}{z+n^{-1}\nu}\|^2 - \tfrac{1}{2}\|z-\falign{x+\epsilon}{z}\|^2\right)
    \]
    satisfies $\lim_{n\to \infty}f_n(\epsilon) = \nu\cdot \nabla_z \tfrac{1}{2}\|z-\falign{x+\epsilon}{z}\|^2$. Furthermore, $f_n$ are all integrable with respect to the Gaussian measure $P$. \cref{lem: lipschitz} with $u=z+\nu/n$, $v=z$ and $w=x+\epsilon$ gives
    \[
        |f_n(\epsilon)|\leq n\|n^{-1}\nu\|\left(\|z + \tfrac{1}{2}n^{-1}\nu\|+\|x+\epsilon\|\right)\leq \|\epsilon\|+\|z\|+\|x\|+\tfrac{1}{2}.
    \]
    Let $g(\epsilon) = \|\epsilon\|+\|z\|+\|x\|+\tfrac{1}{2}$. Then, $|f_n(\epsilon)|\leq g(\epsilon)$ and $g$ is integrable with respect to the Gaussian measure $P$ since the Gaussian measure has finite moments. Lebesgue's dominated convergence theorem now gives
    \[
        \nu\cdot\nabla_z \mathbb{E}\left[\|z-\falign{x+\epsilon}{z}\|^2\right] = \lim_{n\to\infty}\mathbb{E}[f_n(\epsilon)] = \mathbb{E}\left[\lim_{n\to\infty}f_n(\epsilon)\right] = \nu\cdot(z -\falignlim{z}{x}).
    \]
    The proof follows since $\nu$ is an arbitrary direction.
\end{proof}

From~\cref{prop:gradi-formula}, we see that the gradient operator and data limit operator commute when applied to the loss function $\lossdata{\bullet }{\mathcal{X}_N}$. It turns out that many truncated gradient-based optimizers also commute with the data limit operator when applied to the loss function, as the below theorem shows. The proof (see~\cref{apx: optimizer convergence}) is complicated, and the convergence is slow, because the terms of the loss function become correlated over multiple iterations. 

\begin{theorem}\label{thm: iterated convergence full}
    Let $\mathcal{X}_N=\{\xi_n\}_{n=1}^N$ be MRA samples with variance $\tau>0$ based on a signal $x\in \mathbb{R}^L$ as described in~\cref{eq:observed_signals_noshift}. Furthermore, let $a\colon \mathbb{R}^L\times \mathbb{R}^L\to\mathbb{R}^L$ satisfy $\|a(y_1,z_1)-a(y_2,z_2)\|\leq C_a(\|y_1-y_2\|+\|z_1-z_2\|)$ for all $z_1,z_2,y_1,y_2\in\mathbb{R}^L$ and some $C_a>0$. Let $y\in \mathbb{R}^L$ be a non-periodic signal. Let $f_N^k$ and $f^k$ denote $k$ repeated compositions of the functions $f_N(y):=a(y,\faligndata{y}{\mathcal{X}_N})$ and $f(y):=a(y,\falignlim{y}{x})$, respectively. Then, for any $\epsilon\in(0,1)$ and any $k$ and $N>M$ with $M$ sufficiently large, it holds that
    %\begin{equation}
        %\mathbb{P}(\|f_N^k(y)-f^k(y)\|\geq \epsilon)  \leq k\sqrt{e/L}^L\frac{(\log N)^{L/2}}{N\epsilon^2}, \leq C\frac{\rho^{-k}\log N}{N\epsilon^2} + \mathcal{O}\left(k\frac{\log N^{L/2}}{N}\right)
    %\end{equation}
    \begin{equation}
        \mathbb{P}(\|f_N^k(y)-f^k(y)\|\geq \epsilon)  \leq C_1\frac{(\log N)^{L/2}}{N}+ C_2\frac{\log N}{N\epsilon^2},
    \end{equation}
    where $M$, $C_1>0$ and $C_2>0$  depend on $C_a, L, k, y,\tau$ and $x$ but not $\epsilon$ or $N$. In particular, $\lim_{N\to\infty}\mathbb{P}(\|f_N^k(y)-f^k      (y)\|\geq \epsilon)=0$ for any finite $k$ and $\epsilon>0$.
\end{theorem}
\begin{remark}
The iterates generated by our MCA algorithm, which we introduce in section~\ref{sec:method}, uses the same construction as above, with the specific choice
\[
    a(y,z)[n] = (\mathrm{proj}_{\mathcal{M}_{m_2}}z)[n] = \mathcal{F}^{-1}\left\{m_2[k]\hat z[k]/|\hat z[k]|\right\}[n].
\]
This corresponds to block coordinate descent (note that this choice of $a$ only depends on the second argument). This choice of $a$, however, is not Lipschitz due to problems at the origin. For the theory to hold, we could regularize the denominator:
\[
    a(y, z)[n] = (\mathrm{proj}_{\mathcal{M}_{m_2}}z)[n] = \mathcal{F}^{-1}\left\{m_2[k]\hat z[k]/(\delta + |\hat z[k]|)\right\}[n],
\]
for some small $\delta>0$, which has Lipschitz constant bounded by $C_a=\|m_2\|_\infty/\delta$. Other methods can be described on the same form, gradient descent takes the form $a(y,z) = (1-\gamma)y +\gamma z$ where $\gamma>0$ is the step size.   
\end{remark}

One might expect that the same argument used in the proof of \cref{prop:gradi-formula} could be applied to compute the second-order derivative \(\nabla_{z}^2 \losslim{z}{x}\). However, the Lipschitz-like properties of the integrand in the loss function (\cref{lem: lipschitz}) do not carry over to its gradient. In fact, the samples $x+\epsilon$ close to the total discriminant \(\Delta_z\) dominate the computation of \(\nabla_{z}^2 \losslim{z}{x}\). Therefore, when computing \(\nabla_{z}^2 \losslim{z}{x} = \nabla_z(z - \falignlim{z}{x})\), we cannot simply switch order of integration and differentiation and neglect the non-differentiability at points on the discriminant $\Delta_z$. Instead, one would need to differentiate the alignment in a distribution-theoretic sense. As illustrated in \cref{fig:loss-L=5}, examining the level sets of the expected loss function reveals that its Hessian cannot be positive definite for all \(z\), much less be the identity matrix. 

In the following section, we discuss the existence and character of critical points of $\losslim{\bullet}{x}$.

%The smoothness properties of $\mathcal{L}$ combined with an explicit expression for the gradient $\nabla\mathcal{L}$ with respect to $z$ in terms of the aligning function $\overline{\sigma}$ is useful for constructing a gradient-based algorithm that minimizes $\mathcal{L}$. A key property of interest for such methods is the existence and characterization of critical points since local maxima and saddle points can slow down the optimization process. This is the topic of the following section.

\subsection{Critical Points}

Using the expression for $\nabla_{z} \losslim{z}{x}$ provided in \cref{prop:gradi-formula}, and considering the specific geometry of phase manifolds, we can now find the condition under which a signal is critical for $\losslim{\bullet}{x}$.

\begin{proposition}
\label{prop:crit_paral}
    Let $x$ be a non-periodic signal on a phase manifold $\mathcal{M}_{m_2}$, and $\tau>0$. Then $z \in \mathcal{M}_{m_2}$ is critical for $\losslim{\bullet}{x}$ if and only if $\falignlim{z}{x}$ is perpendicular to the tangent space $\mathcal{T}_z \mathcal{M}_{m_2}$. In particular, for $z \in \operatorname{Crit}(\losslim{\bullet}{x})$, every non-zero Fourier mode $\widehat{\falignlim{z}{x}}[k]$ with $k \neq 0$ can exhibit either an in-phase configuration ($0$° phase shift) or an out-of-phase configuration ($180$° phase shift) with respect to $\widehat{z}[k]$.
\end{proposition}
\begin{proof}
    By \cref{prop:gradi-formula}, we have $\nabla_z \losslim{z}{x} = z - \falignlim{z}{x}$. Note that a signal $z \in \operatorname{Crit}(\losslim{\bullet}{x})$ is equivalent to $\nabla_z \losslim{z}{x} \perp \mathcal{T}_z \mathcal{M}_{m_2}$.
    According to \cref{rmk:alpha-torus}, $\mathcal{M}_{m_2}$ is a $d$-dimensional torus if $L$ is odd and the union of two antipodal $d$-dimensional tori if $L$ is even, hence $z$ is perpendicular to $\mathcal{T}_z \mathcal{M}_{m_2}$. So, $z \in \operatorname{Crit}(\losslim{\bullet}{x})$ if and only if $\falignlim{z}{x} \perp \mathcal{T}_z \mathcal{M}_{m_2}$.
\end{proof}

Thus far, we have shown that the expected loss function $\losslim{\bullet}{x}$ is smooth when $\tau>0$ and we computed its gradient. In the next proposition, we introduce a finite set of points on the phase manifold $\mathcal{M}_{m_2}$ such that they belong to $\operatorname{Crit}(\losslim{\bullet}{x})$.

\begin{proposition}
\label{prop:align-same-phase}
    Let $L$ be odd and $x \in \mathbb{R}^L$ be a non-periodic signal on a phase manifold $\mathcal{M}_{m_2}$. Then, a vector $z \in \mathcal{M}_{m_2}$ with $\widehat{z}[k] \in \{\widehat{x}[k],-\widehat{x}[k]\}$ is critical for $\losslim{\bullet}{x}$. In particular, up to circular shift, there are exactly  $2^{\#\supp(\widehat{m}_2)}$ signals on $\mathcal{M}_{m_2}$ that exhibit such a configuration. 
\end{proposition}
\begin{proof}
Without loss of generality, by \cref{lem:invariant-loss}, we can assume that the true signal $x$ has zero Fourier phases. Since $x$ and $z$ differ only by sign flips on the frequencies, we have that $z$ is circularly symmetric, i.e., its Fourier modes are all real, or equivalently, $z$ is invariant under the flipping operator $\flip$. Now, it suffices to show that any such $z\in\mathcal{M}_{m_2}$ is critical for the loss.

By \cref{prop:crit_paral}, we can deduce that $z$ is critical for $\losslim{\bullet}{x}$ if and only if $\falignlim{z}{x}$ is circular symmetric. Since both $x$ and $z$ are invariant under the flipping operator $\flip$, the equivariance property of the averaged aligning function in \cref{lem:invariant-loss} gives
\[
\flip \falignlim{z}{x}
\;=\;
\falignlim{\flip z}{\flip x}
\;=\;
\falignlim{z}{x},
\]
which finishes the proof.
\end{proof}

In our forthcoming conjecture, we discuss the critical points of $\losslim{\bullet}{x}$ and their type. Specifically, we propose that when $L$ is odd, there are no additional critical points for $\losslim{\bullet}{x}$ beside those described in \cref{prop:align-same-phase}.   
\begin{conjecture}
\label{conj}
Let $x \in \mathcal{M}_{m_2}$ be an odd-size signal and $\tau>0$. Then the loss function $\losslim{\bullet}{x}$
    has exactly ${2^{\#\supp(\widehat{m}_2)}}$ critical points up to circular shifts, which are those described in \cref{prop:align-same-phase}. In particular, $x$ and $-x$ correspond to the global minimum and maximum, respectively, of $\losslim{\bullet}{x}$, while the remaining critical points represent saddle points (see \cref{fig:loss-L=5}).
\end{conjecture}

Although \cref{conj} appears natural, proving it universally is challenging. After conducting numerous experiments across various signals with different values of $\tau$ and a sufficient number of samples, we consistently observed that the true signal and its circular shifts are the only local minima of the expected loss function. 

In connection with the conjecture, the following theorem shows that if $\mathcal{M}_{m_2}$ forms a circle (e.g., when $L$ is odd and the true signal $x$ is a sinusoid), then $z=x$ is indeed a local minimum of $\losslim{z}{x}$.
\begin{figure}[ht]
%\begin{wrapfigure}{r}{0.5\textwidth}
    \centering
    \includegraphics[width=0.5\textwidth]{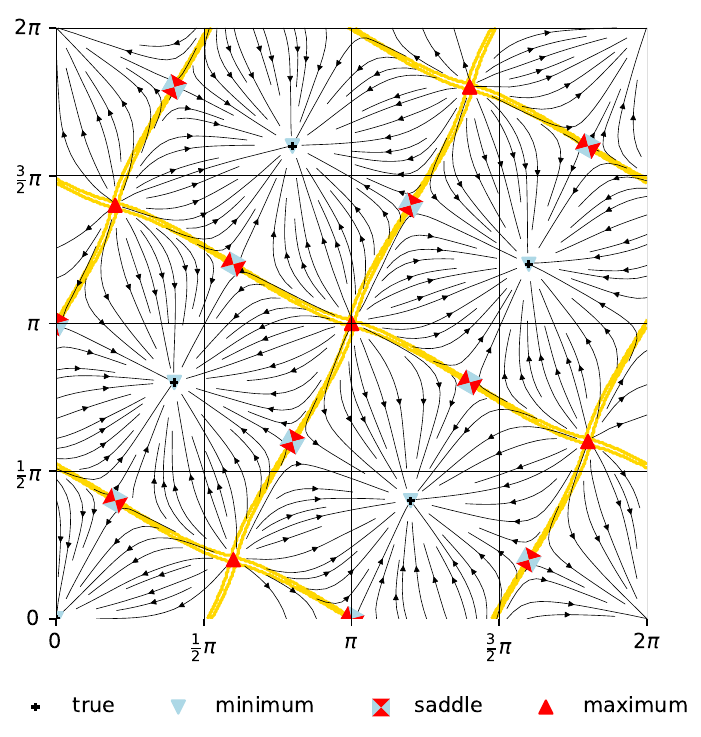}
    \caption{The intersection of the total discriminant of $z = (4/5, -1/5, -1/5, -1/5, -1/5)^\top$ with its corresponding phase manifold is depicted by the gold curve. The gradient $\nabla_z\losslim{z}{x}$ is drawn black, and critical points are marked as minimum, saddle, and maximum, respectively. The signal is marked ``true''.}
    \label{fig:discrim2}
%\end{wrapfigure}
\end{figure}
The significance of this theorem lies in its implication that the MRA problem can be solved for each Fourier mode independently, leading to a finite number of potential solutions for reconstructing the true signal. 
We found this method less stable compared to the simultaneous reconstruction of all modes at once. Additionally, the frequency-wise reconstructions is complicated by having mutually distinct shifts for each frequency. This issue can be solved by finding the shift configuration with minimal loss value, but that makes the reconstruction very costly as the configuration space is extremely large.

\begin{theorem}\label{prop:min special case}
    Let $L$ be odd and $x[n] = c\cos(2\pi k n / L - \phi_0)$ for some integer $k \neq 0$, $\phi \in \mathbb{R}$ and $c > 0$. Then $z=x$ is a local minimum for $\losslim{z}{x}$.
\end{theorem}
\begin{proof}
    By \cref{lem:invariant-loss}, assume that $x$ is circularly symmetric, that is, set \( \phi_0 = 0 \).
 Let $C \in \cso(L)$ 
    such that $||C x-x||= \mathcal{O}(\theta)$ for some $\theta$ that is close to zero. Note that $C$ acts on the $k$th Fourier mode via $U(1)$, which is simply a rotation of order $\mathcal{O}(\theta)$.
    Now, define
    \[\Omega_x=\left \{z \in \mathbb{R}^L \;\middle|\; z[n] = r \cos(\frac{2\pi kn}{L}-\phi) \text{ for } r>0 \text{ and } -\frac{\pi}{L}<\phi<\frac{\pi}{L}\right \}.\]
    This positive cone consists of all sinusoids of frequency $k$ that are closer to $x$ than any other $\sigma_i(x)$. Similarly, we have $\Omega_{C x} = C \Omega_x$, which possesses the same property for $C x$. Let $\Omega = \Omega_x \cap \Omega_{C x} \cap \Omega_{C^\top x}$, then we have that $\xi = x + \epsilon \in \Omega$ is closest to $x$ and $Cx$ among all their shifts if and only if $\flip \xi = \flip(x + \epsilon) = x + \flip \epsilon$   is also closest to $x$ and $Cx$ among all their shifts. Now, we claim 
\begin{align*}
    & ||x - \falign{x}{\xi}||^2 + ||x - \falign{x}{\flip \xi}||^2 
    < \ ||Cx - \falign{Cx}{\xi}||^2 + ||Cx - \falign{Cx}{\flip \xi}||^2,
\end{align*} 
which by the earlier explanation is equivalent to
\begin{equation}
    \label{eq:C_theta_x}
     ||x - \xi||^2 + ||x - \flip \xi||^2 = 2||\epsilon||^2 
    <  ||Cx - \xi||^2 + ||Cx - \flip \xi||^2.
\end{equation}
    The right-hand side of \cref{eq:C_theta_x} can be further simplified as follows: 
     \begin{align*}
         2\left( ||Cx||^2 + ||\xi||^2-\langle C x,\xi  \rangle -\langle C x,\flip \xi  \rangle \right) &= 4 \left( ||x||^2 + \langle x,\epsilon\rangle-\langle C x,x + \tfrac{1}{2}(\epsilon+\flip \epsilon) \rangle \right) +2||\epsilon||^2\\
         &=4 \left \langle x-C x,x+\tfrac{1}{2}(\epsilon+\flip \epsilon)\right \rangle+2||\epsilon||^2.
     \end{align*}
     Therefore it is enough to show that $\langle C x - x, x + \frac{\epsilon + \flip \epsilon}{2} \rangle < 0$. This follows from the fact that the $k$th Fourier phase of $x + \tfrac{\epsilon + \flip \epsilon}{2}$ is zero since $\xi$ and $\flip \xi$ are closest to $x$, which has zero phase at frequency $k$, and the phase manifold is a circle. In contrast, the $k$th Fourier coefficient $e^{ik\theta} - 1$ of $C x - x$ has a negative real part, so the inner product must be negative.
     
Note that \( \epsilon \) and \( \flip{\epsilon} \) have the same probability density, and the function \( \losslim{z}{x} \) can be interpreted as the expected value of the expression \( \tfrac{1}{2}||z-\falign{z}{x+\epsilon}||^2 + \tfrac{1}{2}||z-\falign{z}{x+\flip{\epsilon}}||^2 \), where \( x+\epsilon \) is closer to \( z \) than other shifts of \( z \).
Furthermore, when \( \theta \) is close to zero, as in the proof of \cref{prop:gradi-formula}, the computation of both \( \losslim{x}{x} \) and \( \losslim{Cx}{x} \) can be restricted to considering only the samples \( \xi = x+\epsilon \) within \( \Omega \). Combining this with \cref{eq:C_theta_x}, we confirm that \( \losslim{x}{x} < \losslim{Cx}{x} \). Hence, \( z = x \) is a local minimum for \( \losslim{z}{x} \).
\end{proof}

\subsection{Low-SNR Regime -- Aligning with Noise}
\label{sec:low_snr}
 In this section, we analyze the averaged aligning function when $\tau \to \infty$, that is when the SNR tends to zero. In the low-SNR regime, the behavior of $\falignlim{z}{x}$ is akin to that of $\falignlimgen{1}{z}{\mathbf{0}_L}$. This relationship can be understood by applying \cref{lem:properties_L_sigma}, yielding
\[ \falignlim{z}{x} = \tau \cdot \falignlimgen{1}{\tfrac{z}{\tau}}{\tfrac{x}{\tau}} = \tau \cdot \falignlimgen{1}{z}{\tfrac{x}{\tau}} \approx \tau \cdot \falignlimgen{1}{z}{\mathbf{0}_L}, \]
as $\tau \gg 0$. 

We now show that, in this regime, aligning a template with white noise yields a signal that preserves the Fourier phases of the template. The phenomenon is studied using statistical tools in \cite{balanov2024einstein}. Our proof, in contrast, is geometric and considerably simpler.
\begin{theorem}
   Let \(z \in \mathbb{R}^L\) be a signal with \(\widehat{z}[k] \neq 0\) for all \(k > 0\) and let \(\tau > 0\). Then \(\falignlim{z}{\mathbf{0}_L}\) has the same Fourier phases as \(z\) and has zero mean.
\end{theorem}
\begin{proof}
    According to \cref{lem:properties_L_sigma}, we have $m_1(\falignlim{z}{\mathbf{0}_L}) = m_1(\mathbf{0}_L) = 0$. Next, we aim to show that $\falignlim{z}{\mathbf{0}_L}$ has the same Fourier phases as $z$. 
 By virtue of \cref{lem:invariant-loss}, we may assume that $z$ exhibits zero Fourier phase, and also without loss of generality, we can take $\widehat{z}[0]=0$. Moreover, upon alignment with $z$, the noisy samples $\xi = \mathbf{0}_L+ \epsilon = \epsilon$ can be regarded as vectors $\sigma(z; \epsilon)$ within a positive polyhedral cone specified by:
\[
\operatorname{Cone}^{+}(z) = \{y \in \mathbb{R}^L \mid \langle y - \sigma_k(y), z \rangle \geq 0 \text{ for every } k\},
\]
that is, the cone of signals closer to $z$ compared to any shifted version of that signal.
Since $z$ has zero Fourier phase we deduce that $\sigma(z; \epsilon) \in \operatorname{Cone}^{+}(z)$ if and only if $\flip \sigma(z; \epsilon) \in \operatorname{Cone}^{+}(z)$. Note that the probability of observing $\epsilon$ and $\flip \epsilon$ is the same, and since $\sigma(z; \epsilon) +  \sigma(z; \flip\epsilon)$ is circularly symmetric, we conclude that $\falignlim{z}{\mathbf{0}_L}$ is also circularly symmetric, i.e., its Fourier phases are either $0^\circ$ or $180^\circ$.

We now claim that the latter is not possible; thus $\falignlim{z}{\mathbf{0}_L}$ has the same Fourier phases as $z$. To see this, let $y \in \operatorname{Cone}^{+}(z)$ and define a modified signal $w$ by setting its Fourier coefficients as
\[
\widehat{w}[k] = \operatorname{sgn}(\operatorname{Re}(\widehat{y}[k])) \cdot \widehat{y}[k], \quad \text{for all } k.
\] 
By construction, $w$ is also in $\operatorname{Cone}^{+}(z)$ and moreover, due to symmetry in noise distribution, $y$ and $w$ are equally likely to appear. Finally, averaging such signals would yield strictly nonnegative phases and thus finishes the proof. 
\end{proof}

\section{Reconstruction algorithm}\label{sec:method}
From an implementation viewpoint, the proposed MRA method is one of constrained non-linear optimization. In this section, it is helpful to separate out the unknown shifts $\pmb{r}$, which are otherwise considered latent variables of the loss function. Namely, given some data set $\mathcal{X}_N$, a template $z\in\mathbb{R}^L$ and unknown shifts $\pmb{r}=(r_n)_{n=1}^N\in (\mathbb{Z}_L)^N$, we define the joint loss function $\lossdata{z,\pmb r}{\mathcal{X}_N}$ as follows:
\begin{equation}
     \lossdata{z,\pmb{r}}{\mathcal{X}_N}=\frac{1}{N}\sum_{n=1}^N \|z-\sigma_{r_n}\xi_n\|^2.\label{eq: joint loss}
\end{equation}
Now, constraining the MLE formulation in \cref{eq: mle} to the phase manifold $\mathcal{M}_{m_2^*}$ where we take $m_2^*$ to be the estimate $m_2^*(\mathcal{X}_N)$ given by \cref{eq:m2 estimate}, results in an optimization problem over the signals $z$ and the shifts $\pmb{r}$:
\begin{equation}\label{eq: joint optimization}
    \min\left\{\lossdata{z,\pmb{r}}{\mathcal{X}_N}\;\middle|\; z\in \mathcal{M}_{m_2^*}, \; \pmb{r}\in (\mathbb{Z}_L)^N\right\}\\
\end{equation}
Choosing $r_n(z) = \arg\min_{r}\|z-\sigma_{r}(\xi_n)\|$ shows that \cref{eq: joint optimization} is equivalent to minimizing $\lossdata{z}{\mathcal{X}_N}$ with respect to $z$, constrained to $\mathcal{M}_{m_2^*}$. We will now present a block coordinate descent algorithm for minimizing~\cref{eq: joint loss} that alternates between solving two proxy problems -- finding the shifts and finding the signal. We use this approach to construct a sequence $\{(z_i, \pmb{r}_i)\}_{i=0}^I$ of estimates:
\begin{equation}
    \pmb{r}_{i+1} = \argmin_{\pmb{r}\in (\mathbb{Z}_L)^N}\lossdata{z_i, \pmb{r}}{\mathcal{X}_N},\qquad z_{i+1} = \argmin_{z\in\mathcal{M}_{m_2^*}}\lossdata{z,\pmb{r}_{i+1}}{\mathcal{X}_N}\label{eq: block coordinate descent}
\end{equation}
It follows by construction that the loss values of the sequence converge.
\begin{proposition}(Convergence of block coordinate descent)\label{prop: coordinate convergence}
    Define $\{(z_i, \pmb{r}_i)\}_{i=1}^\infty$ as in \cref{eq: block coordinate descent} and let $\ell_i^N := \lossdata{z_i,\pmb r_i}{\mathcal{X}_N}$. Then, $\ell_i^N$ converges to some $\ell_\infty^N$ as $i\to\infty$.
\end{proposition}
\begin{proof}
    By the construction of $(z_i,\pmb r_i)$, it holds that
    \[
        \ell_{i+1}^N = \lossdata{z_{i+1}, \pmb{r}_{i+1}}{\mathcal{X}_N}\leq \lossdata{z_{i},\pmb{r}_{i+1}}{\mathcal{X}_N}\leq \lossdata{z_{i},\pmb{r}_{i}}{\mathcal{X}_N} = \ell_i^N.
    \]
    Furthermore, $\ell_i^N \geq 0$ for all $i$. Hence, the sequence $(\ell_i^N)_{i=0}^N$ is decreasing and bounded from below and therefore, it has a limit.
\end{proof}

\subsection{Efficient Computations}
We now turn to the practical aspects of the block coordinate descent algorithm presented above. First, we show that the subproblems have tractable solutions. Consider the subproblem of optimizing~\cref{eq: joint loss} subject to fixed shifts $\pmb{r}\in(\mathbb{Z}_L)^N$. The subproblem can be solved by a simple computation in the Fourier domain:
\begin{align*}
    \lossdata{z,\pmb r}{\mathcal{X}_N} &= \frac{1}{2N}\sum_{n=1}^N \sum_{k=1}^L |\hat z[k] - \widehat{(\sigma_{r_n}\xi_n)}[k]|^2 = \sum_{k=1}^L\frac{1}{2N}\sum_{n=1}^N  |\hat z[k] - \widehat{(\sigma_{r_n}\xi_n)}[k] |^2\\
    &= \sum_{k=1}^L \tfrac{1}{2}m_2
    [k]^2 + \tfrac{1}{2N}\sum_{n=1}^N|\widehat{\xi_n}[k]|^2 - \Re\left({\tfrac{1}{N}\sum_{n=1}^N\widehat{(\sigma_{r_n}\xi_n)}[k]}\overline{\hat z[k]}\right)=\\
 &= - \sum_{k=1}^L \Re\left({\tfrac{1}{N}\sum_{n=1}^N\widehat{(\sigma_{r_n}\xi_n)}[k]}\overline{\hat z[k]}\right) + \mathrm{const.}
\end{align*}
We can view the above as $L$ separate optimization problems, one for each Fourier coefficient. The loss functions are given by inner products, and each coefficient $\hat z_k$ is constrained to the circle of radius $m_k$. By Cauchy--Schwarz, the optimal solution is when $\hat z_k$ is aligned to the $k$-th Fourier mode of $\tfrac{1}{N}\sum_{n=1}^N\sigma_{r_n}\xi_n$ on the circle, resulting in a minimizer
\begin{equation}
    \argmin_{z\in\mathcal{M}_{m_2^*}}\lossdata{z,\pmb{r}}{\mathcal{X}_N} = \mathrm{proj}_{\mathcal{M}_{m_2}^*}\left(\frac{1}{N}\sum_{n=1}^N{\sigma_{r_n}\xi_n}\right), %m_2[k]\frac{\frac{1}{N}\sum_{n=1}^N(\widehat{\sigma_{r_n}\xi_n})[k]}{\left|\frac{1}{N}\sum_{n=1}^N(\widehat{\sigma_{r_n}\xi_n})[k]\right|}.
    \label{eq: signal opt}
\end{equation}
where we use the projection operator, defined by its Fourier transform:
\[
    \widehat{\mathrm{proj}}_{\mathcal{M}_{m_2^*}\left(z\right)}[k] = \begin{cases}
        m_2^*[k]\frac{\hat z [k]}{|\hat z [k]|}, & \quad \hat z[k]\neq 0\\
        m_2^*[k], &\quad \text{otherwise}.
    \end{cases}
\]
Next, we minimize~\cref{eq: joint loss} with a fixed signal $z\in \mathcal{M}_{m_2^*}$. With $\pmb r^*=(r_n^*)_{n=1}^N$ defined by:
\begin{equation}
    \pmb{r}^* = \argmin_{\pmb r\in (\mathbb{Z}_L)^N} \lossdata{z,\pmb{r}}{\mathcal{X}_N}, \quad \text{then}\quad r_n^*= \arg\min_{r\in\mathbb{Z}_L}\|z - \sigma_r \xi_n\|.\label{eq: shift opt}
\end{equation}
The aligning function can be efficiently computed in two steps. First, we find the maximal argument $r_n$ \cref{eq:find_shift} in the cross-correlation between $z$ and $\xi_n$:
\[
    r_n^* = \argmax_{r\in \mathbb{Z}_L} \|\sigma_r\xi_n - z\|^2 = \argmax_{r\in \mathbb{Z}_L} (\flip\xi_n \ast z)[-r] = \argmax_{r\in \{1,\dots,L\}} \mathcal{F}^{-1}[\hat \xi_n\odot \overline{\hat z}][r],
\]
which is conveniently expressed as the inverse Fourier transform of the Hadamard (element-wise) product $(\hat x \odot \overline{\hat z})[\ell]=\hat x[\ell]\overline{\hat z[\ell]}$. The computational complexity of a single step is $\mathcal{O}(N L\log L)$: first $\mathcal{O}(N L\log L)$ for the cross-correlations, $\mathcal{O}(NL)$ for finding the maximal arguments, and $\mathcal{O}(NL)$ for averaging the signals. Each of these operations are trivially parallelizable over the data. \cref{alg:fixed_point} shows a detailed description including computational specifics. 

It follows from a composition of~\cref{eq: signal opt} and \cref{eq: shift opt} that $z_{i+1} = \mathrm{proj}_{\mathcal{M}_{m_2^*}}(\faligndata{z_i}{\mathcal{X}_N})$. This formulation hides the optimization over shifts, which better reflects the philosophy of our approach. Namely, our method works because the shifts are estimated correctly on average, whether individual shifts are estimated correctly is irrelevant. \cref{prop: coordinate convergence} can be trivially adapted to this new formulation:
\begin{corollary}
    Let $(z_i)_{i=0}^N$ be a sequence defined by $z_{i+1}=\mathrm{proj}_{\mathcal{M}_{m_2^*}}(\faligndata{z_i}{\mathcal{X}_N})$, that is aligning the signal to $\mathcal{X}_N$ and projecting back to the manifold $\mathcal{M}_{m_2}$. Then, $\ell_i^N := \lossdata{z_i}{\mathcal{X}_N}$ converges as $i\to \infty$.
\end{corollary}

%\begin{remark}
%The use of higher-order methods is complicated by the loss function, whose Hessian is the identity matrix almost everywhere. In contrast, the Hessian of the loss function in the infinite data regime is not the identity (i.e, $\lim_{N\to\infty}\nabla^2\lossdata{\bullet}{\mathcal{X}_N}\neq \nabla^2 \lim_{N\to\infty}\lossdata{\bullet}{\mathcal{X}_N}$). One way to estimate the Hessian is to compute finite differences with $\mathcal{O}(N^{-\alpha})$ increments, which is costly. Alternatively, one can derive analytical formulas by implicit differentiation on continuous shifts but unfortunately, the resulting expressions contain third-order moments, which are both costly to compute and inaccurate due to high variance. 
%\end{remark}

\begin{algorithm}[t]
\caption{Moment-constrained alignment (MCA)}
\label{alg:fixed_point}
\begin{algorithmic}[1]
    \Require Input value: $\delta, \tau, \mathcal{X}_N=\{\xi_1,\ldots, \xi_N\}$ \Comment{}{$\delta$ is tolerance and $\xi_n \in \mathbb{R}^L$}
    \Ensure Output value: \texttt{Reconstructed signal}
    \For{$n=1,\dots,N$}
    \State $\widehat \xi_n \gets \texttt{FFT}(\xi_n)$\Comment{Discrete Fourier transform}
    \EndFor
    \For{$k=1,\dots,L$}
    \State $\widehat{m}_2[k] \gets\max( 
          \frac{1}{N} \sum_{n=1}^N |\widehat \xi_n[k]|^2 - \tau^2, 0)$
          \Comment{Second-order moment}
    \EndFor
    \State $\widehat{z}_0 \gets\widehat{\xi}_0$
    \State $m \gets 0$
    \While{$|\lossdata{z_m}{\mathcal{X}_N}-\lossdata{z_{m-1}}{\mathcal{X}_N}|> \delta $}
        \Comment{Stop if we hit a fixed point}
        \State $m\gets m+1$
        \For{$n=1,\dots, N$}\Comment{Compute optimal shifts}
            \State $\mathrm{corr}_n \gets \texttt{invFFT}((\widehat \xi_n[k]\overline{\widehat{z}_{m-1}[k]})_{k=1}^L)$\Comment{Inverse discrete Fourier transform}
            \State $r_n \gets \argmax_{\ell=1,\dots,L}\mathrm{corr}_n[\ell]$ 
        \EndFor
        \For{$k=1,\dots,L$}
        \State $\widehat{z}_{m}^*[k] = \frac{1}{N}\sum_{n=1}^N \widehat{\sigma_{r_n}(\xi_n)}[k]$\Comment{Compute aligning function $\faligndata{z_{m-1}}{\mathcal{X}_N}$}
        \State $\widehat{z}_m[k] \gets \widehat{m}_2[k]\widehat{z}_{m}^*[k]/|\widehat{z}_{m}^*[k]|$  \Comment{Projection $\mathrm{proj}_{\mathcal{M}_{m_2}}(\overline{\sigma})$ to the manifold}  
        \EndFor
    \EndWhile\label{euclidendwhile}
    \State \textbf{return} $z_m$\Comment{The reconstructed signal}
\end{algorithmic}
\end{algorithm}

%The connection is more intuitive with the following formulation:
%\begin{gather}
%    \min\left\{\losslim{z, \pmb{r}}{x}\quad \middle| z\in \mathcal{M}_{m_2}, \;\;\pmb{r}\colon \mathbb{R}^L \to \mathbb{Z}_L\right\}, \\
%    \losslim{z,\pmb{r}}{x} = \int_{\mathbb{R}^L} \|z -\sigma_{\pmb{r}(y)}(y)\|^2 \exp(-\|x-y\|^2/(2\tau^2))\mathrm{d}y
%\end{gather}
%Similarly to the finite data setting, we define a sequence $\{(z_i, \pmb{r}_i)\}_{i=0}^I$ of estimates:
%\[
   % \pmb{r}_{i+1} = \argmin_{\pmb{r}\colon \mathbb{R}^L\to \mathbb{Z}_L}\losslim{z_i, \pmb{r}}{x},\qquad z_{i+1} = \argmin_{z\in\mathcal{M}_{m_2}}\losslim{z,\pmb{r}_{i+1}}{x}\label{eq: block coordinate descent}
%\]
%And similarly to the finite data setting, the optimal choices have analytic expressions, $z_{i+1} = \mathrm{proj}$

\section{Experiments}
\label{sec:experiments}

%\begin{figure}[ht]
    %\captionsetup{width=0.45\textwidth}
    %\centering
    %\begin{minipage}{.5\textwidth}
    %\centering
    %\includegraphics[width=\linewidth]{error/error.png}
   % \caption{Normalized root-mean-square error for square signal, averaged over 40 runs with $N=10^4$ samples and varying $\sigma$. The 90\% range of outcomes is shaded.}
    %\label{fig:error N=1e4}
   % \end{minipage}%
   % \begin{minipage}{.5\textwidth}
    %    \centering
    %    \includegraphics[width=\linewidth]{error/compute.png}
    %    \caption{Computation time for square signal, averaged over 40 runs with $N=10^4$ samples and varying $\tau$. The 90\% range of outcomes is shaded.}
     %   \label{fig:time N=1e4}
    %\end{minipage}
   % \vspace{-2em}
%\end{figure}

\begin{figure}[ht]
    \captionsetup{width=0.45\textwidth}
    \centering
    \begin{minipage}{.5\textwidth}
        \centering
        \begin{tikzpicture}
            \node[anchor=north west] at (0,0) {\includegraphics[width=\linewidth]{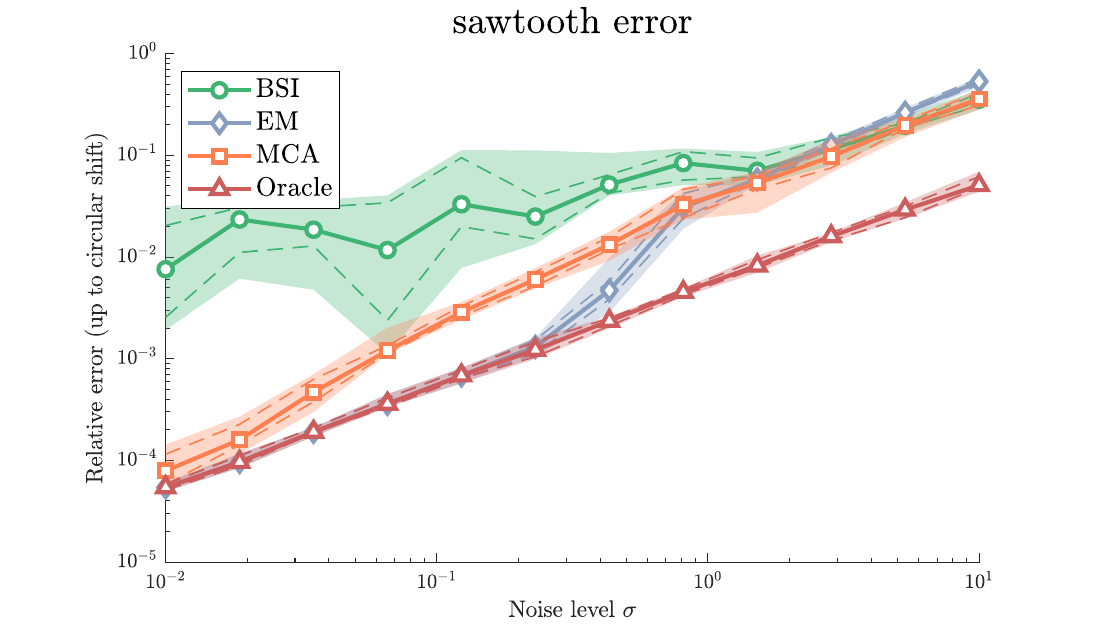}};
            \node[anchor=north west] at (4.5cm,-2.5cm) {\includegraphics[width=2.5cm]{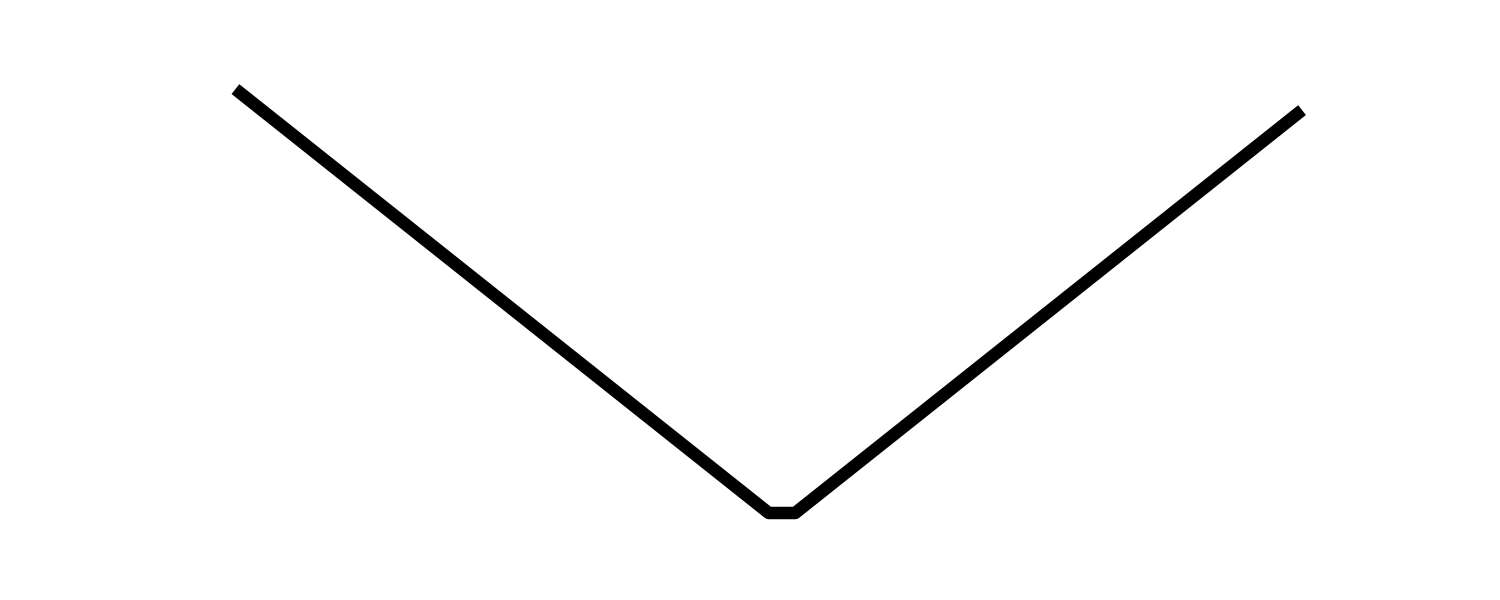}};
        \end{tikzpicture}
        
    \end{minipage}%
    \begin{minipage}{.5\textwidth}
        \centering
        \includegraphics[width=\linewidth]{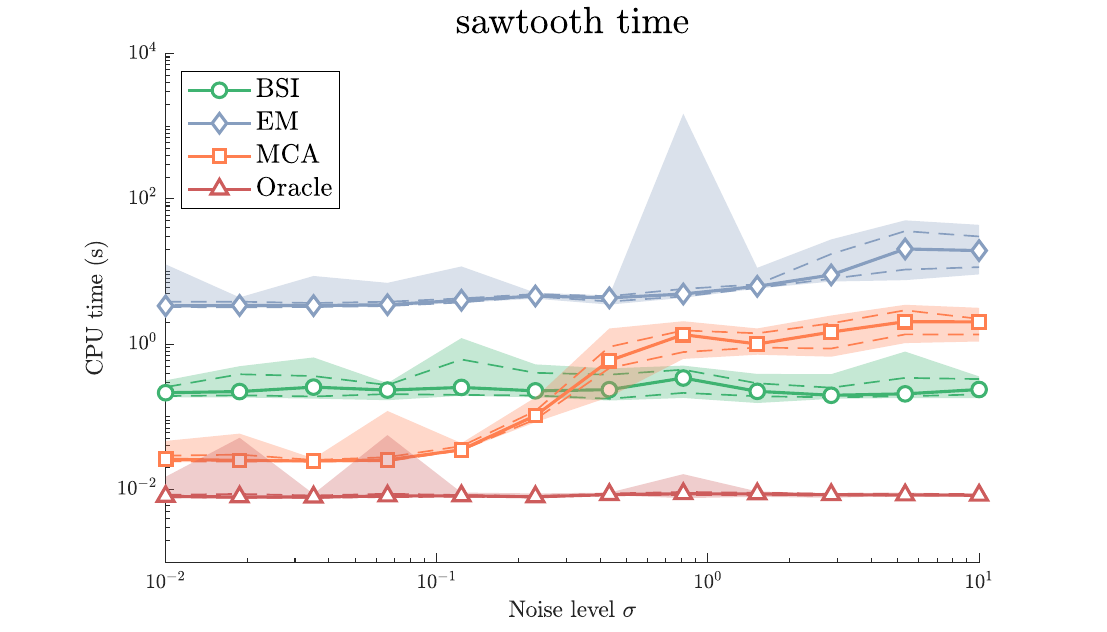}
    \end{minipage}
    \begin{minipage}{.5\textwidth}
        \centering
        \begin{tikzpicture}
            \node[anchor=north west] at (0,0) {\includegraphics[width=\linewidth]{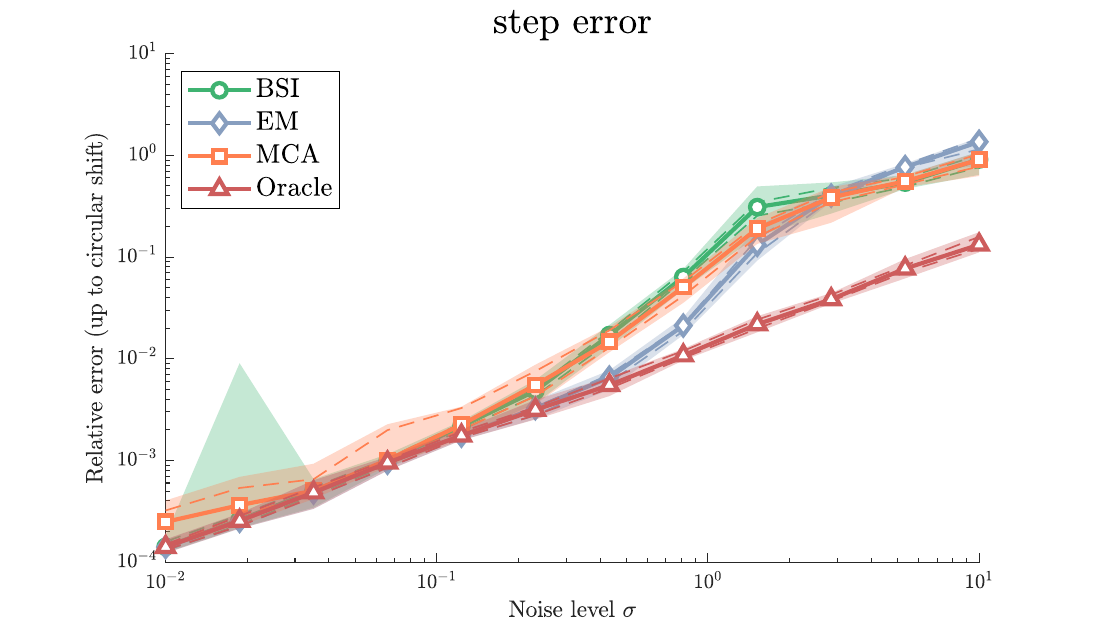}};
            \node[anchor=north west] at (4.5cm,-2.5cm) {\includegraphics[width=2.5cm]{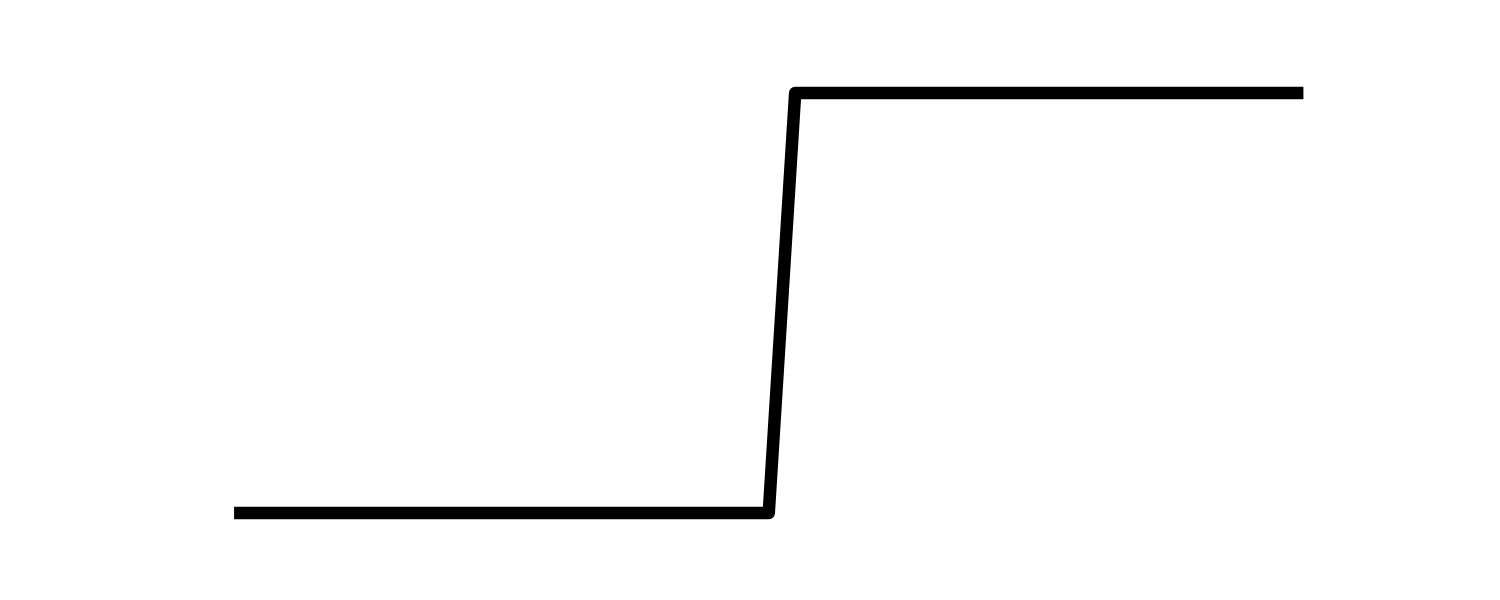}};
        \end{tikzpicture}
    \end{minipage}%
    \begin{minipage}{.5\textwidth}
        \centering
        \includegraphics[width=\linewidth]{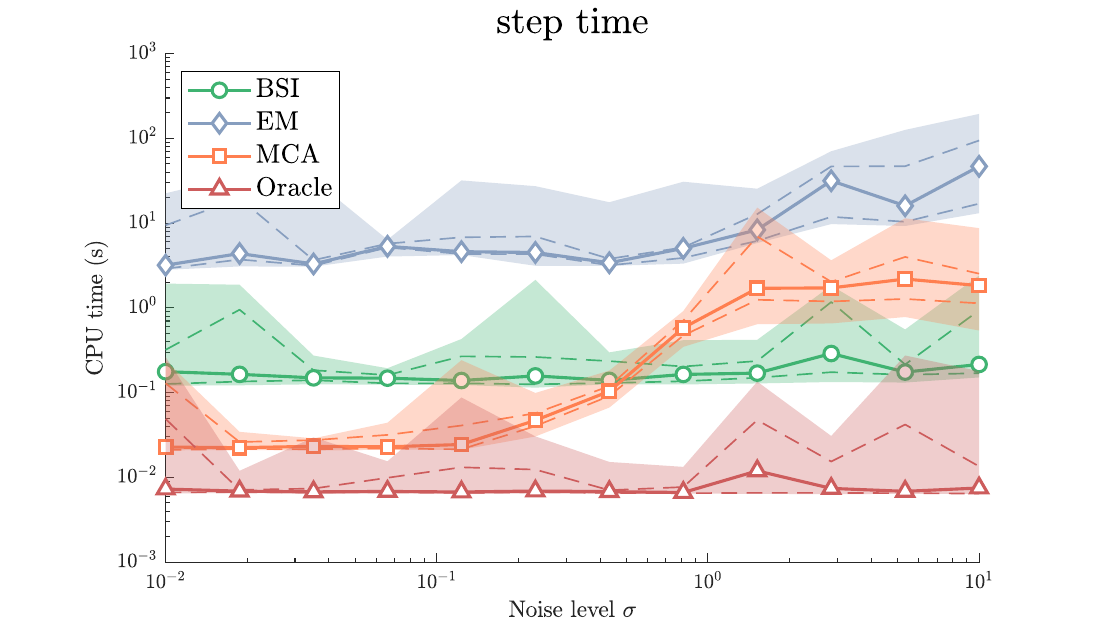}
    \end{minipage}
    \begin{minipage}{.5\textwidth}
        \centering
        \begin{tikzpicture}
            \node[anchor=north west] at (0,0) {\includegraphics[width=\linewidth]{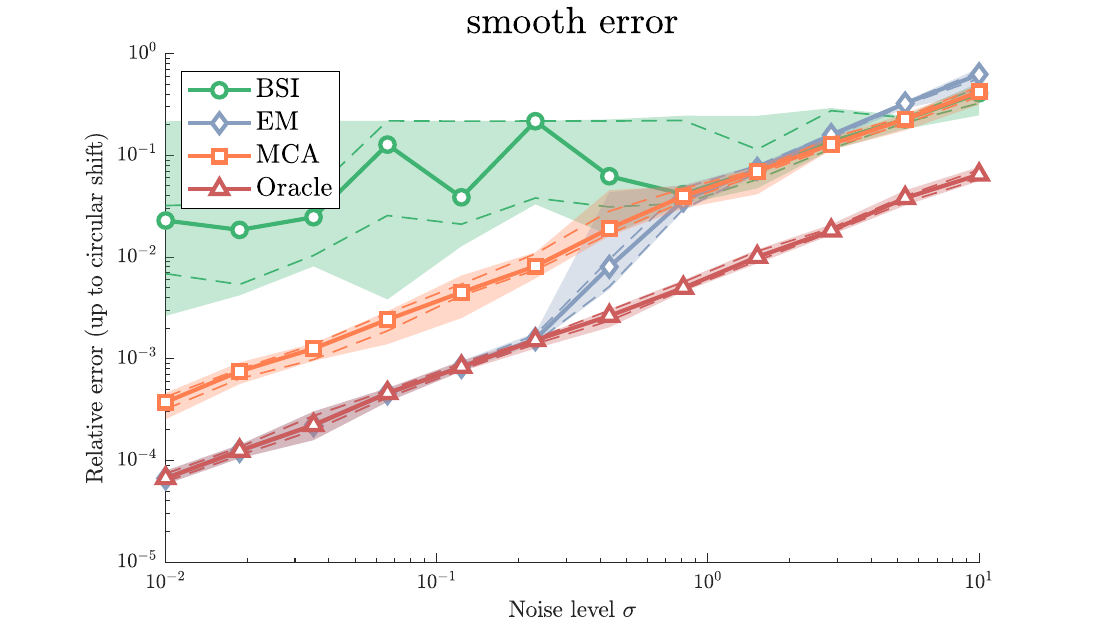}};
            \node[anchor=north west] at (4.5cm,-2.5cm) {\includegraphics[width=2.5cm]{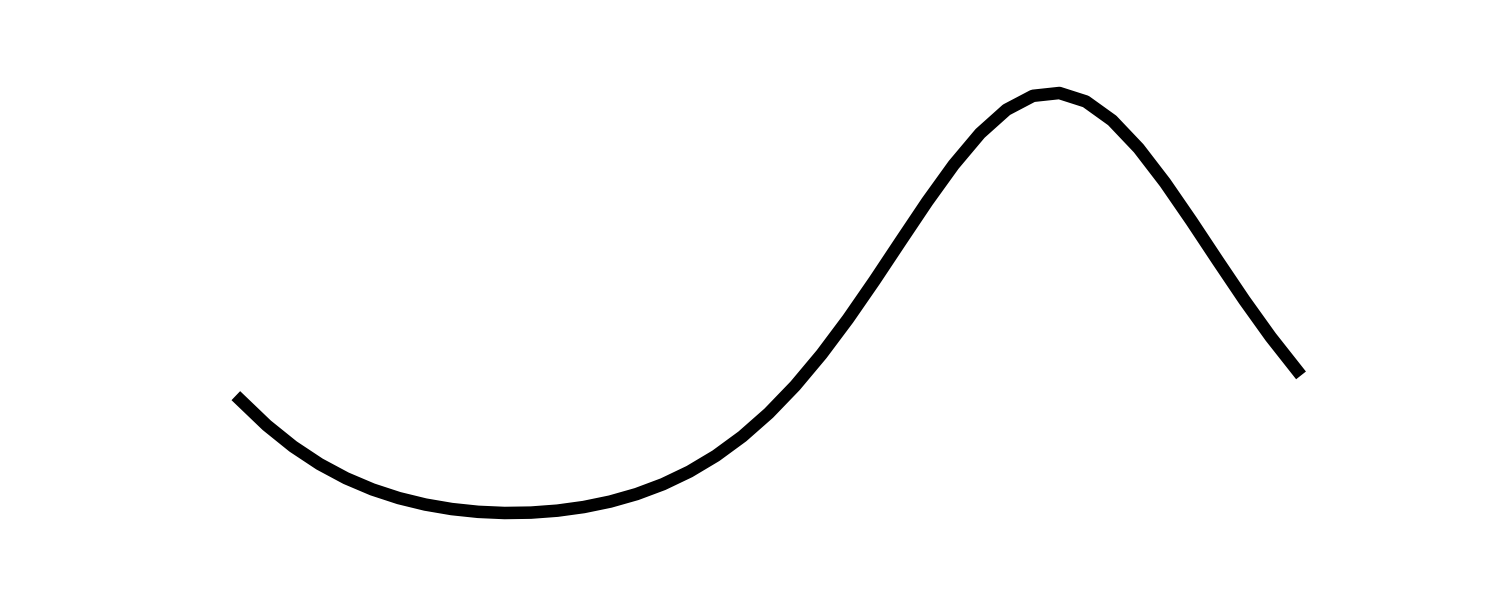}};
        \end{tikzpicture}
        \caption{Normalized root-mean-square error for sawtooth, square and smooth signal, averaged over 40 runs with $N=10^4$ samples and varying $\sigma$. The 100\% range is shaded and the 60\% range of outcomes is shown as dashed lines.}
        \label{fig:error N=1e4}
    \end{minipage}%
    \begin{minipage}{.5\textwidth}
        \centering
        \includegraphics[width=\linewidth]{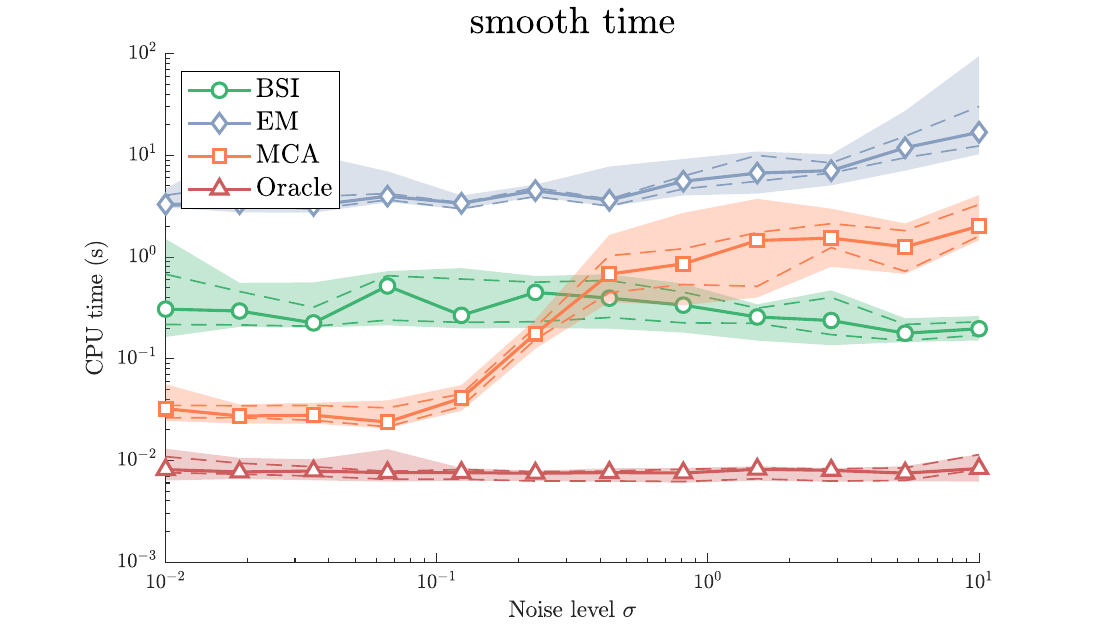}
        \caption{Computation time for sawtooth, square and smooth signal, averaged over 40 runs with $N=10^4$ samples and varying $\tau$. The 100\% range is shaded and the 60\% range of outcomes is shown as dashed lines.}
        \label{fig:time N=1e4}
    \end{minipage}
    \vspace{-2em}
\end{figure}

In this section, we present implementation details of our method and evaluate its numerical performance against existing results. This section closely follows the setup from~\cite{bendory2017bispectrum}. In the following experiments, we take $x$ as a vector of length $L=41$ or $21$ depending on the experiment. We look at three different signals: the square wave or ``step'' function $x_{step}$, the sawtooth wave $x_{saw}$, and a smooth signal $x_{smooth}$. For all $k=1,\dots, L$, we set
\begin{align*}
    x_{step}[k] &= \begin{cases}
        1, & k\geq \lfloor L/2 \rfloor\\
        0, & \mathrm{otherwise}
    \end{cases}\\
    x_{saw}[k] &= |k - L/2|\\
    x_{smooth}[k] &= \exp(\sin(2\pi k/L))
\end{align*}
The $N$ samples are generated according to~\cref{eq:observ-sig}. We measure the quality of a reconstruction $z$ using the normalized root-mean-square error of $z$ aligned to $x$:
\[
    \mathrm{NRMSE}(z,x)= \frac{\|z-\falign{z}{x}\|}{\|x\|}=\min_{r\in \{1,\dots,L\}}\frac{\|z - \sigma_{r}(x)\|}{\|x\|}.
\]

To compare our proposed \cref{alg:fixed_point}, we also implemented the EM algorithm, the bispectrum inversion using phase synchronization from Bendory et al.~\cite{bendory2017bispectrum}, and a known-shifts oracle that averages over the samples, perfectly aligned according to the true shifts. The phase synchronization algorithm was originally implemented in Bendory et al.~\cite{bendory2017bispectrum} with a Riemannian trust region method -- we use the Matlab code provided in that paper to evaluate the bispectrum method. For phase synchronization in the bispectrum inversion method, we solve the interior optimization problem with power iterations, as detailed in~\cite{bendory2017bispectrum}. We run EM with a warm start of 3000 iterations on a batch of 1000 samples as in~\cite{bendory2017bispectrum}, and then iterate with the full data until convergence.  %For the Python version of bispectrum inversion, we chose the phase synchronisation approach from~\cite{bendory2017bispectrum} with power iterations instead of Riemannian descent. The method was compared against the Matlab version to make sure the two implementations matched up. For~\cref{fig:error N=1e4}, we noticed that the Python implementation of Bispectrum method struggles in the low noise regime compared to the other approaches. Consequently,~\cref{fig:error N=1e4} is generated entirely from Matlab code, to make sure the most efficient version of bispectrum inversion was used for comparison. 

All methods were implemented in \texttt{jax}, a Python package for scientific computing. It uses XLA (accelerated linear algebra) which allows for efficient computations, and is capable of under-the-hood multi-threading through batched operations. None of the code was explicitly parallelized, but \texttt{jax} automatically batches over the samples and executes all at once~\cite{jax2018github}. Although \texttt{jax} supports automatic differentiation, we compute all gradients explicitly to reduce computational costs. The code is available at \href{https://github.com/emastr/moment-constrained-alignment}{https://github.com/emastr/moment-constrained-alignment}. The comparisons in~\cref{fig:error N=1e4} were made using the original Matlab implementations of bispectrum inversion and EM from~\cite{bendory2017bispectrum} as well as our own Matlab implementation of MCA. This figure uses Riemannian optimization for the phase synchronization step in bispectrum inversion, instead of the power iterations.

The experiments were all run on a CPU with 10 cores, 2 out of which run at 4.7 MHz max clock speed, and 8 that run at 3.6 MHz.  We set the tolerance of all methods to $10^{-6}$, regardless of noise level. 

\cref{fig:error N=1e4,fig:time N=1e4} depict NRMSE and computation time for $N=10^4$ samples and three different signals. We averaged the results over 40 runs and indicate the range of outcomes as a shaded interval. For the step signal, MCA outperforms bispectrum inversion in terms of NRMSE, particularly within the range $\tau\in (0.3, 2)$. In the same interval, EM outperforms both methods. For higher noise, $\tau>2$, bispectrum and constrained alignment match up, and slightly outperform EM. For the sawtooth and the smooth signals, the bispectrum method fails to converge at low noise levels $\tau<0.5$, whereas our method converges smoothly towards the true signal as $\tau\to 0$.

\begin{figure}[t]
    \centering
    \includegraphics[width=\textwidth]{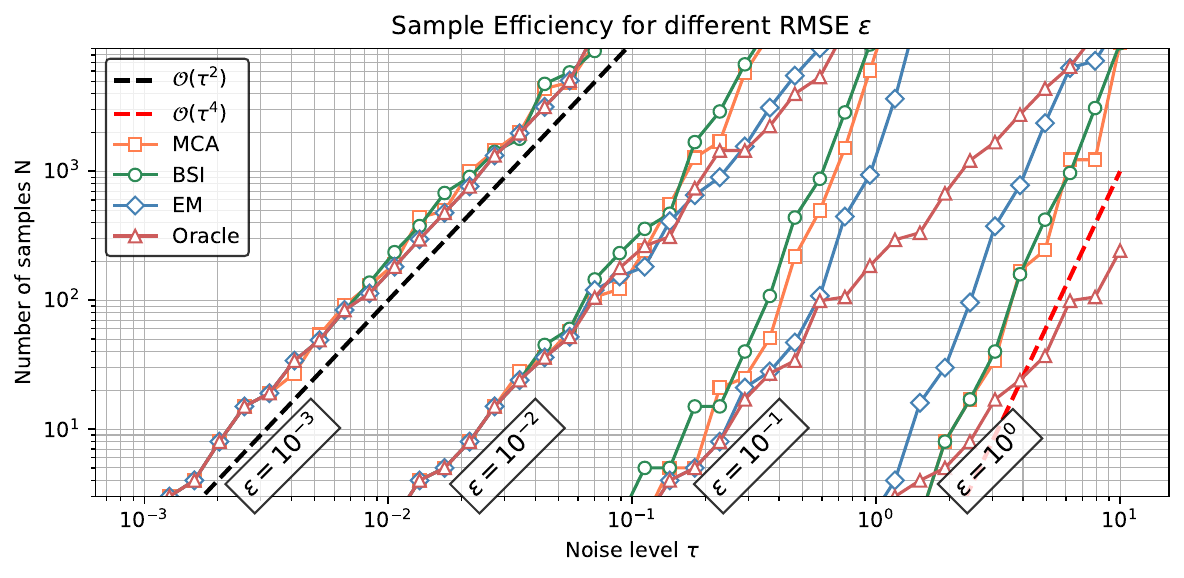}
    \caption{Sample efficiency plots for the different methods. We show the number of samples required to reach a given error $\epsilon$ as a function of the noise level $\tau$.}
    \label{fig:sample eff}
    \vspace{-2em}
\end{figure}
In terms of computational speed, MCA tends to be faster than bispectrum inversion by about an order of magnitude for low noise (less than $0.1$), largely due to the high overhead incurred from computing the bispectrum (and EM is even slower). However, this cost is independent of $\tau$, which explains why bispectrum is faster than MCA in the high-noise regime (greater than $1.0$). The latter converges in two iterations at low noise, but the number of iterations required increases over the interval $(0.1, 1.0)$ and plateaus at 200 iterations. One way to control this cost could be to use a noise-dependent tolerance, similar to Morozov discrepancy \cite{morozov}.

\cref{fig:sample eff} shows the number of samples required to reach different error levels $\epsilon$, as a function of noise $\tau$. For low noise levels, we observe a scaling of $N=\mathcal{O}(\tau^2)$. This is consistent with the theory at low SNR. For high noise (greater than $1.0$), we observe worse scaling, $N=\mathcal{O}(\tau^4)$, whereas the theory developed by Perry et al.~\cite{singer2019samplecomplex} predicts $\mathcal{O}(\tau^6)$ scaling. \cref{fig:sample eff} shows that the efficiency curves are increasing in slope for higher SNR. Sixth-order scaling emerges at higher sample numbers and lower SNR than what was possible on our machine.

We lastly comment on the ability of MCA to avoid critical points. In \cref{prop:align-same-phase}, we identified $2^{\lfloor L/2\rfloor}$ critical points that could inhibit convergence of MCA to the true signal. \Cref{fig:crit loss vs err} shows the distances between the MCA reconstruction and each possible critical point, plotted against the loss value of that critical point. For $10^3$ samples, MCA is not significantly closer to the true signal than other critical points. As the sample size increases to $10^4$, it starts to separate out, and at $10^6$ samples, the true signal is closer to the reconstruction than the remaining critical points by an order of magnitude. Furthermore, the loss value of the true signal is clearly the smallest out of the critical points.

\begin{figure}
    \centering
    \includegraphics[width=\linewidth]{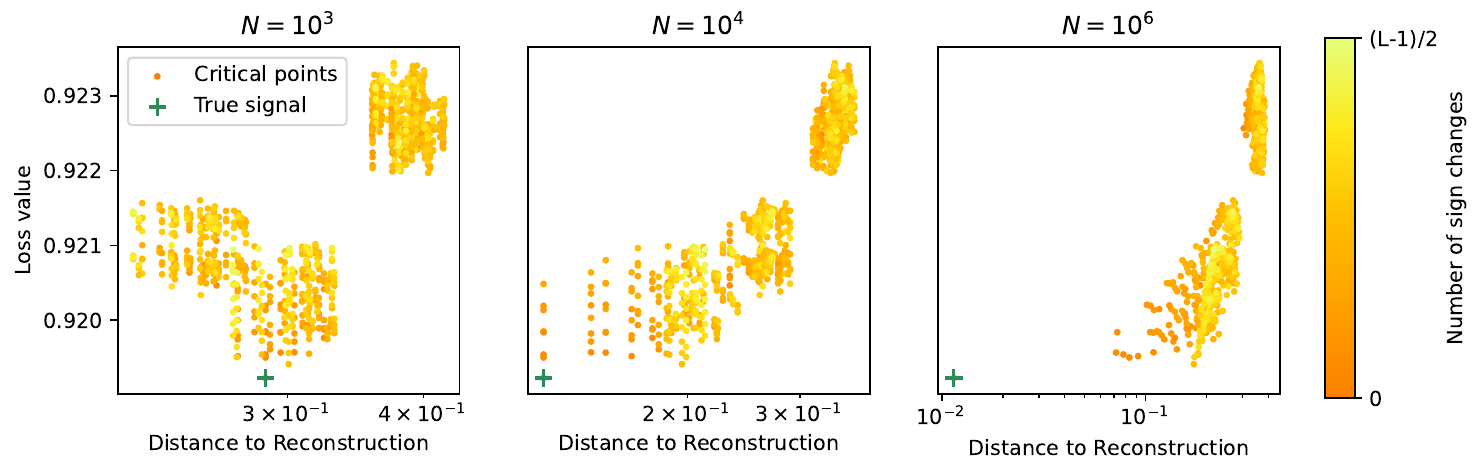}
    \caption{Scatter plots of data $\{(d_n,\ell_n)\}_{n=1}^{1024}$ showing the loss values $\ell_n = \losslim{z_n}{x}$, plotted against the distance $d_n = \|x_N-z_n\|$ from $z_n$ to an MCA reconstruction $x_N$ of the signal using $N=10^3,10^4$ and $10^6$ data samples. The points $\{z_n\}_{n=1}^{1024}$ are exactly the set of points such that $\hat z_n[k]\in \{\hat x[k],-\hat x[k]\}$ for $k=1,\dots,\lfloor L/2\rfloor$, i.e. the critical points identified in~\cref{prop:align-same-phase}. The points are colored according to the size of the set $\{k\colon \hat z_n[k]=-\hat x[k]\}$, or the number of $180$ degree phase flips from the true signal. With a signal length $L=21$ of mean 0, there are exactly $2^{10}=1024$ such combinations. The noise level $\tau$ is $1.6$, and $x$ is the square wave signal.}
    \label{fig:crit loss vs err}
\end{figure}

\section{Conclusion and Outlook}
In this work, we introduced a novel approach to multireference alignment based on minimizing a loss function restricted to the manifold of signals with a fixed power spectrum. We proved that the loss converges uniformly to a smooth function in the infinite-data limit and related its gradient to template alignment. We characterized a subset of the critical points of the averaged loss function in the infinite-data limit, showing that the true signal is one of them. Furthermore, we showed that for a signal with one non-zero Fourier mode, the true signal is also a local minimizer. %\textcolor{red}{Finally, we conjectured a precise count and position of saddle points, minima, and maxima.}

We then introduced a provably convergent reconstruction algorithm, MCA. We showed that the MCA iterations on the finite data loss converge to the corresponding iterations taken on the infinite data loss in the limit, and compared its performance to EM and bispectrum inversion. Our method achieves higher accuracy than bispectrum inversion overall. Additionally, it is faster than EM across all noise levels and also faster than bispectrum inversion in the low-noise regime. Moreover, in the high-noise regime, our method outperforms EM in terms of accuracy.

There are two main directions to take for future research on MCA. First, it remains to show consistency for generic signals. That is, that  the loss function of MCA has the underlying signal as its global minimum. We have presented several representations of the loss, but it is still unclear which of these representations could help in a consistency proof.

The second direction of research is to adapt MCA to cryo-EM. Much like MRA, cryo-EM has a set of second-order invariant statistics (see Kam's method, explained in~\cite{kam_method_2019}). Restricting the search space to the set of signals that match these statistics defines an analog to the phase manifold. Similarly, the aligning function has an analog in cryo-EM as well~\cite{barnett2017rapidsolutioncryoemreconstruction}. Finding a way to correctly combine these already existing concepts the key obstacle to implementing MCA for cryo-EM.

\bibliographystyle{siamplain}
\bibliography{literature}

\clearpage
\appendix
\section{Proofs from \cref{sec:phase-recovery}}\label{apx: lipschitz}
\begin{proof}[Proof of~\cref{lem: lipschitz}]
    Define the interpolating function $w(\alpha)=(1-\alpha)u + \alpha v$. Then, let
    \[
        \varphi(\alpha) = \tfrac{1}{2}\|w(\alpha)-\falign{u}{w}\|^2 - \tfrac{1}{2}\|w(\alpha)-\falign{v}{w}\|^2.
    \]
    Then, $\varphi(\alpha)$ is continuous on $\alpha\in (0,1)$, and
    \begin{align*}
        \varphi(0) &= \tfrac{1}{2}\|u-\falign{u}{w}\|^2 - \tfrac{1}{2}\|u-\falign{v}{w}\|^2\leq 0,\\
        \varphi(1) &= \tfrac{1}{2}\|v-\falign{u}{w}\|^2 - \tfrac{1}{2}\|v-\falign{v}{w}\|^2\geq 0.
    \end{align*}
    By the intermediate value theorem, there is a $\alpha^*\in [0,1]$ and a corresponding $w^* = w(\alpha^*)$ such that $\varphi(\alpha^*)=0$, meaning $\|w^*-\falign{u}{w}\|=\|w^*-\falign{v}{w}\|$. Therefore, we have the following:
    \begin{align*}
        &\left| \tfrac{1}{2}\|u-\falign{u}{w}\|^2-\tfrac{1}{2}\|v-\falign{v}{w}\|^2 \right|=\\
        &=\left|\tfrac{1}{2}\|u-\falign{u}{w}\|^2-\tfrac{1}{2}\|w^*-\falign{u}{w}\|^2+\tfrac{1}{2}\|w^*-\falign{v}{w}\|^2-\tfrac{1}{2}\|v-\falign{v}{w}\|^2\right|\\
        &= \left|\tfrac{1}{2}\|u\|^2-\tfrac{1}{2}\|v\|^2-\langle u-w^*,\falign{u}{w}\rangle-\langle w^*-v,\falign{v}{w}\rangle\right|\\
        &= \left|\tfrac{1}{2}\langle u-v,u+v\rangle -\langle u-w^*,\falign{u}{w}\rangle-\langle w^*-v,\falign{v}{w}\rangle\right|\\
        &\leq \tfrac{1}{2}\|u-v\|\|u+v\|+(\|u-w^*\|+\|v-w^*\|)\|w\|\\
        &= \|u-v\|(\tfrac{1}{2}\|u+v\|+\|w\|),
    \end{align*}
    where we used Cauchy--Schwarz and $\|\falign{a}{b}\|=\|b\|$ to bound the inner products, and used that $\|u-w^*\|+\|w^*-v\|=\|u-v\|$ because $w^*$ is on the line between $u$ and $v$. 
\end{proof}

\begin{proof}[Proof of~\cref{prop: lipschitz}]
    By \cref{lem: lipschitz} with $w = x + \epsilon=\xi$ where $\|x\|=\|u\|=\|v\| = \|m_2\|$ combined with the triangle inequality:
    \[
        \|u-\falign{u}{\xi}\|^2- \|v-\falign{v}{\xi}\|^2 \leq \|u-v\|(4\|m_2\| + 2\|\epsilon\|).
    \]
    Summing over the terms of the loss and applying the above termwise finally gives:
    \[
        |\lossdata{u}{\mathcal{X}_N}-\lossdata{v}{\mathcal{X}_N}|\leq \frac{1}{N}\sum_{n=1}^N\|u-v\|(2\|m_2\| + \|\epsilon_n\|) = \left(2\|m_2\| + \frac{1}{N}\sum_{n=1}^N\|\epsilon_n\|\right)\|u-v\|.
    \]
    We identify the Lipschitz constant $C_N$ in the last expression. Lastly, since $C_N$ is a constant $2\|m_2\|$ plus a sample mean of independent, $\chi(L)$-distributed random variables, $C_N$ is bounded away from its expectation by the Chebyshev inequality:
    \[
        \mathbb{P}(|C_N - 2\|m_2\|-\sqrt{2}\tau\tfrac{\Gamma((L+1)/2)}{\Gamma(L/2)}|\geq \epsilon)\leq \frac{\mathrm{Var}(C_N)}{\epsilon^2} = \frac{\mathrm{Var}(\|\epsilon_1\|)}{N\epsilon^2} \leq \frac{\tau^2 L}{N\epsilon^2}.
    \]
    %$C_N$ converges almost surely by the strong law of large numbers:
    %\[
    %    C_N \overset{\small a.s.}{\rightarrow} 2m + \mathbb{E}[\|\epsilon_1\|]=2m+\sqrt{2}\tau\frac{\Gamma((L+1)/2)}{\Gamma(L/2)}
    %\]
\end{proof}

\subsection{Proof of~\cref{prop: lipschitz}}

\begin{proof}
    We take help by the following result. Let $(\mathbb{P},\Omega,\mathcal{F})$ be the probability space from which the data set $\mathcal{X}_N$ is drawn. Let $C_N$ denote the Lipschitz constant from \cref{prop: lipschitz}, and define $C_\infty:=\lim_{N\to\infty}C_N$. Let $\{z_k\}_{k=1}^M$ denote a discretization of $\mathcal{M}$ resolution $h$, meaning $\min_{k}\|z_k-z\|\leq h$ for all $z\in \mathcal{M}$. Such a discretization can be obtained with $M=(\|m_2\|_\infty/h)^L$ points. Furthermore, define the sets
    \[
         A_{k,N} := \left\{\omega\in \Omega\colon |\lossdata{z_k}{\mathcal{X}_N}-\losslim{z_k}{x}|\geq \epsilon/2\right\},\quad\text{for}\quad k=1,2,\dots.
    \]
    Also define $B_N = \{\omega\in\Omega\colon |C_N-C_\infty|\geq C_\infty\}$. By \cref{prop: lipschitz}, $\mathbb{P}(B_N)\leq \tfrac{\tau^2 L}{N C_\infty^2}$. Now, consider the set intersection $\tilde A_N=B_N^c\cap \bigcap_{k=1}^M A_{k,N}^c$. Suppose $\mathcal{X}_N$ is drawn from $\tilde A_N$, and consider any $z\in\mathcal{M}$. Then, there is a $k$ such that $\|z-z_k\|\leq h$, with $h =\epsilon/(2\cdot 3C_\infty)$, and since $C_N\leq 2C_\infty$,
    \begin{multline*}
        |\lossdata{z}{\mathcal{X}_N}-\losslim{z}{x}|\leq \\
        \leq |\lossdata{z}{\mathcal{X}_N}-\lossdata{z_k}{\mathcal{X}_N}|+|\lossdata{z_k}{\mathcal{X}_N}-\losslim{z_k}{x}|+|\losslim{z_k}{x}-\losslim{z}{x}|\\
        \leq 3C_\infty\|z-z_k\|+|\lossdata{z_k}{\mathcal{X}_N}-\losslim{z_k}{x}| \leq \epsilon/2 + \epsilon/2 = \epsilon,
    \end{multline*}
    where we used that $A_{k,N}^c\subset \tilde A_N$ in the last step. Since the above is true for any $z\in\mathcal{M}$, 
    \[
         \mathbb{P}\left\{\sup_{z\in\mathcal{M}}|\lossdata{z}{\mathcal{X}_N} - \losslim{z}{x}|\leq \epsilon\right\} \geq \mathbb{P}(\tilde A) \geq 1-\mathbb{P}(B_N)-\sum_{k=1}^M\mathbb{P}(A_{k,N}).
    \]
    Flipping the inequality therefore gives
    \begin{align*}
        \mathbb{P}\left\{\sup_{z\in\mathcal{M}}|\lossdata{z}{\mathcal{X}_N} - \losslim{z}{x}|\geq\epsilon\right\} &\leq \mathbb{P}(B_N)+\sum_{k=1}^M\mathbb{P}(A_{k,N})\\&\leq \frac{\tau^2L}{NC_\infty^2}+\frac{(4\|m_2\|_\infty^2+2\tau^2L)\|m_2\|_\infty^L}{\epsilon^{L+2} N}.
    \end{align*}
    Specifically, choosing $\epsilon/C_\infty<1$ and $\epsilon/\|m_2\|_\infty<1$ gives the desired result.
    
\end{proof}

%\section{Supplementary Material}

\section{Proof of \cref{thm: iterated convergence full}}
\label{apx: optimizer convergence}
We begin with some definitions.
\begin{setting}\label{setting: alignment}
Let $(\Omega, \mathcal{F}, \mathbb{P})$ be a probability space. In this section we fix $\tau>0$ and $x\in\mathbb{R}^L$. Let $\mathcal{X}_N=\{\xi_n\}_{n=1}^N$ be a set of measurements, where each $\xi_n=x+\epsilon_n$ and $\epsilon_n\colon\Omega\to  \mathbb{R}^L$ are iid. $\mathcal{N}(0,\tau)$ distributed random variables on $\Omega$. To shorten notation, let
\[
    f_N(y) = a(y, \faligndata{y}{\mathcal{X}_N}),\qquad f(y) = a(y, \falignlim{y}{x})
\]
for any $y\in\mathbb{R}^L$, where $a\colon\mathbb{R}^L\times \mathbb{R}^L \to \mathbb{R}^L$ is a Lipschitz-continuous map in the sense that there exists a $C_a>0$ so that for any $y_1,y_2,z_1,z_2$ in $\mathbb{R}^L$,
\begin{equation}\label{eq: lipschitz two args}
    \|a(y_1, z_1)-a(y_2,z_2)\|\leq C_a\left(\|y_1-y_2\|+\|z_1-z_2\|\right).
\end{equation}
\end{setting}
We denote by $g^k$ the functional composition $g\circ g\circ \dots\circ g$ of $g$ $k$ times.  In this section we investigate under which circumstances the following holds:
\begin{equation}\label{eq: probable convergence}
    \lim_{N\to \infty}\mathbb{P}(\|f^k_N(y)-f^k(y)\|\geq \epsilon) = 0,
\end{equation}
for finite $k$. We will aim to find an upper bound for the rate at which the above converges. Specifically, such a bound would serve to motivate the study of first-order optimizers on the expected loss $f$. The end goal of this section is to prove~\cref{thm: iterated convergence full}, which we restate here in a more compact format using~\cref{setting: alignment}:
\begin{theorem}\label{thm: iterated convergence}
    Let $f_N,f$ be as in~\cref{setting: alignment} and $y\in \mathbb{R}^L$ a nonperiodic signal. Then, for any $\epsilon\in(0,1)$, any $k$ and $N>M$:
    \begin{equation}
        \mathbb{P}(\|f_N^k(y)-f^k(y)\|\geq \epsilon)  \leq C_1\frac{(\log N)^{L/2}}{N} + C_2\frac{\log N}{N\epsilon^2},
    \end{equation}
     where $M, C_1>0$ and $C_2>0$ depend on $C_a, L, k, y$ and $f$ but not $\epsilon$ or $N$. In particular, $\lim_{N\to\infty}\mathbb{P}(\|f_N^k(y)-f^k      (y)\|\geq \epsilon)=0$ for any $k$ and $\epsilon>0$.
\end{theorem}
We begin by introducing a useful ordering principle. Namely, for any two random variables $x,y\colon\Omega \to B$ of some space $B$ with an ordering $\leq$, it holds that if $x\leq y$, then for any $a\in B$, $\{x\geq a\}\subseteq\{y\geq a\}$, and as a consequence,
\begin{equation}\label{eq: upper bound of rv}
    \text{if } x\leq y,\quad\text{then}\quad\mathbb{P}(x\geq a)\leq \mathbb{P}(y\geq a)\quad \forall a\in B.
\end{equation}
A consequence of \cref{eq: upper bound of rv} is that probabilities that a sum of variables exceeds some threshold, is itself bounded above by the sum of the probabilities that its terms are. Namely, by the triangle inequality $\|x_1 + x_2 +\dots x_M\|\geq\delta$ implies that at least one of $\|x_m\|\geq \delta/M$ must hold for $x_m\in\mathbb{R}^L$ and $m=1,\dots,M$. In particular, this means that
\[
    \{\|x_1+\dots+x_M\|\geq \delta\}\subseteq\{\|x_1\|\geq \tfrac{\delta}{M}\}\cup\cdots\cup\{\|x_M\|\geq\tfrac{\delta}{M}\}
\]
for random variables $x_1,\dots,x_M\in\mathbb{R}^L$. By subadditivity of the probability measure:
\begin{equation}\label{eq: triangle inequality}
    \mathbb{P}(\|x_1+\dots +x_M\|\geq \delta)\leq \mathbb{P}(\|x_1\|\geq \tfrac{\delta}{M})+\dots+\mathbb{P}(\|x_M\|\geq \tfrac{\delta}{M}).
\end{equation}
We proceed to prove the following estimate:
\begin{lemma}\label{lemma: random bound}
    Let $(\Omega,\mathcal{F},\mathbb{P})$ be a probability space. Let $g$ and $f\colon \Omega\times B\to B$ be random functions and $x\colon\Omega\to B$ a random variable on a Banach space $B$ with norm $\|\cdot\|$. Let $a\colon B\times B\to B$ be a Lipschitz function as in~\cref{eq: lipschitz two args} with constant $C_a>0$. Then, for any $\epsilon>0$, $0<\delta\leq\epsilon/(3C_a)$ and $y\in B$ the following inequality holds:
    \begin{multline*}
        \mathbb{P}(\|a(x,g(x))-a(y,f(y))\|\geq \epsilon) \\  \leq \mathbb{P}(\|x- y\|\geq \delta) + \mathbb{P}(\|g(y)-f(y)\|\geq \tfrac{\epsilon}{3C_a})+\mathbb{P}\left(\sup_{\|z-y\|\leq \delta} \!\!\|g(z)-g(y)\|\geq \tfrac{\epsilon}{3C_a}\right).
    \end{multline*}
\end{lemma}
\begin{proof}
The proof will follow from a suitable partition of the outcome $\|a(x,g(x))-a(y,f(y))\|\geq \epsilon$. \Cref{eq: lipschitz two args} gives
\begin{align*}
    \|a(x,g(x))-a(y,f(y))\|&\leq C_a(\|x-y\|+\|g(x)-f(y)\|)\\
    &\leq C_a(\|x-y\|+\|g(x)-g(y)\|+\|g(y)-f(y)\|).
\end{align*}
We use the above with~\cref{eq: triangle inequality}, which gives
\begin{multline}\label{eq: inclusion eq}
    \{\|a(x,g(x))-a(y,f(y))\|\geq \epsilon\}\\
    \subseteq \{\|x-y\|\geq \tfrac{\epsilon}{3C_a}\}\cup\{\|g(x)-g(y)\|\geq \tfrac{\epsilon}{3C_a}\}\cup \{\|g(y)-f(y)\|\geq \tfrac{\epsilon}{3C_a}\}.
\end{multline}
Next, for any two sets $A,B$ it holds that $A= (A\cap B)\cup (A\cap B^c)\subseteq (A\cap B)\cup B^c$, which we apply to the sets $A=\{\|g(x)-g(y)\|\geq \epsilon(3C_a)^{-1}\}$ and $B=\{\|x-y\|\leq \delta\}$ for any $\delta>0$:
\begin{align*}
    \{\|g(x)-g(y)\|\geq \tfrac{\epsilon}{3C_a}\}&\subseteq\left(\{\|g(x)-g(y)\|\geq \tfrac{\epsilon}{3C_a}\}\cap\{\|x-y\|\leq\delta\}\right)\cup\{\|x-y\|\geq \delta\}\\
    &\subseteq\left\{\sup_{\|z-y\|\leq \delta}\|g(z)-g(y)\|\geq\tfrac{\epsilon}{3C_a}\right\}\cup\{\|x-y\|\geq\delta\}.
\end{align*} 
In the last step we used that $\{h(y)\geq \epsilon\}\cap\{y\in A\}\subseteq\{\sup_{z\in A} h(z)\geq \epsilon\}$ for any random variable $y$, function $h$ and random set $A$.
Inserted into~\cref{eq: inclusion eq} combined with subadditivity of the probability measure:
\begin{align*}
    \mathbb{P}&(\|a(x,g(x))-f(y,g(y))\|\geq \epsilon)\\&\leq \mathbb{P}(\|x-y\|\geq \delta)+\mathbb{P}(\|g(y)-f(y)\|\geq \tfrac{\epsilon}{3C_a})+\mathbb{P}\left(\sup_{\|z-y\|\leq\delta}\|g(z)-g(y)\|\geq \tfrac{\epsilon}{3C_a}\right)
\end{align*}
where we used $\{\|x-y\|\geq \delta\}\cup\{\|x-y\|\geq\tfrac{\epsilon}{3C_a}\}=\{\|x-y\|\geq\delta\}$ since $\delta\leq\tfrac{\epsilon}{3C_a}$.
\end{proof}
We will now use \cref{lemma: random bound} to obtain a bound on the convergence rate of~\cref{eq: probable convergence}. First, we define some notation:
\begin{align}
    p_{N,k}(\epsilon) &:= \mathbb{P}(\|f_N^k(y)-f^k(y)\|\geq \epsilon),\label{eq:pN}\\
    q_N(y,\epsilon)&:=\mathbb{P}(\|f_N(y)-f(y)\|\geq \epsilon),\label{eq:qN}\\
    r_N(y,\epsilon,\delta) &:= \mathbb{P}\left(\sup_{\|z-y\|\leq \delta}\|f_N(z)-f_N(y)\|\geq\epsilon\right)\label{eq:rN}.
\end{align}
Then~\cref{lemma: random bound} with $g,f,x,y$ and $B$ set to $f_N,f,f_N^{k-1}(y),f^{k-1}(y)$ and $\mathbb{R}^L$ from~\cref{setting: alignment}, respectively, gives a recursive bound on $p_{N,k}(\epsilon)$:
\begin{equation}\label{eq: inequality recursion}
    p_{N,k}(\epsilon) \leq p_{N,k-1}(\delta) + q_N(f^{k-1}(y), \tfrac{\epsilon}{3C_a}) + r_{N}(f^{k-1}(y), \tfrac{\epsilon}{3C_a}, \delta),\qquad \text{for any} \quad \delta\leq \tfrac{\epsilon}{3C_a}.
\end{equation}
The first term $p_{N,k-1}(\delta)$ is the probability of $f^{k-1}(y)$ deviating from $f_N(y)^{k-1}$ by more than $\delta$, and it allows us to argue inductively, since $p_{N,0}(\epsilon)=\mathbb{P}(\|y-y\|\geq \epsilon)=0$ for all $\epsilon>0$. The second term $q_N(y,\epsilon)$, is trivially bounded by Chebyshev's inequality:
\begin{equation}\label{eq: chebyshev qn bound}
    q_N(y,\epsilon)=\mathbb{P}(\|f_N(y)-f(y)\|\geq\epsilon)\leq \frac{\mathrm{Var}(\falign{y}{\xi_1})}{\epsilon^2N}.
\end{equation}
We now seek to bound the terms $r_N(y,\epsilon,\delta)$ for arbitrary $y$ and $\epsilon$, and some unknown $\delta$. Throughout these estimates we will make assumptions on the size of $\epsilon$, which will all be independent on $N$ (possibly dependent on $k$). We will also estimate $\delta$ depending on $N$, which will require a gradually contracting error $\delta$ for smaller $k$. 

\subsection{Estimation of $r_N(y,\epsilon,\delta)$} Estimating the term $r_N(y,\epsilon,\delta)$ is complicated. By definition, $f_N(z) = \tfrac{1}{N}\sum_{n=1}^N\falign{z}{\xi_n}$ with $\xi_n = x+\epsilon_n$. Hence, $f_N$ is a sum of piecewise constant functions. Roughly, taking the supremum of $\|f_N(y)-f_N(z)\|$ with $z$ in a local neighborhood around $y$ amounts to finding $z$ has the most number of discontinuities in the terms of $f_N$  between itself and $y$. Therefore, bounding $q_N$ will ultimately depend on the probability of finding a discontinuity in this neighborhood. We start by proving a helpful lemma:
\begin{lemma}\label{lemma: discontinuous sum}
    Let $(\Omega,\mathcal{F},\mathbb{P})$ be a probability space. Suppose $g_n\colon \Omega\times \mathbb{R}^L\to \mathbb{R}^L$ are iid. piecewise constant functions, uniformly bounded in $L^2$-norm by some constant $C>0$. Let $\Delta_n\colon\Omega\to \mathrm{Set}(\mathbb{R}^L)$ denote the random set of points at which $g_n$ is discontinuous. Let $\overline{g}_N(y):=\frac{1}{N}\sum_{n=1}^Ng_n(y)$. Then,
    \[
        %\mathbb{P}\left(\sup_{z\in B_\delta(y)}\|\overline{g}_N(y)-\overline{g}_N(z)\|\geq \epsilon\right) \leq \exp\left(-N D_{\mathrm{KL}}\left(\mathbb{P}(\Delta_1\cap B_\delta(y)\neq \emptyset)\|(2C)^{-1}\epsilon\right)\right),
        \mathbb{P}\left(\sup_{z\in B_\delta(y)}\|\overline{g}_N(y)-\overline{g}_N(z)\|\geq \epsilon\right) \leq \frac{4C^2\mathbb{P}(\Delta_1\cap B_\delta(y)\neq \emptyset)}{\epsilon^2 N}
    \]
    %where $D_\mathrm{KL}(p\|q)=q\log \{q/p\} + (1-q)+log\{(1-q)/(1-p)\}$ and $B_\delta(y)=\{z\colon \|z-y\|\leq \delta\}$ is the ball of radius $\delta$ centered at $y$.
    where $B_\delta(y)=\{z\colon \|z-y\|\leq \delta\}$ is the ball of radius $\delta$ centered at $y$.
\end{lemma}
\begin{proof}
We make use of~\cref{eq: upper bound of rv} with the following chain of inequalities:
\[
    \sup_{z\in B_\delta(y)}\left\|\sum_{n=1}^Ng_n(z)-g_n(y)\right\|\leq \sup_{z\in B_\delta(y)}\sum_{n=1}^N\|g_n(z)-g_n(y)\|\leq \sum_{n=1}^N\sup_{z\in B_\delta(y)}\|g_n(z)-g_n(y)\|
\]
to write
\begin{equation}\label{eq: termwise bound and sup}
  \mathbb{P}\left(\sup_{z\in B_\delta(y)}\left\|\sum_{n=1}^Ng_n(z)-g_n(y)\right\|\geq N\epsilon\right) \leq \mathbb{P}\left(\sum_{n=1}^N\sup_{z\in B_\delta(y)}\|g_n(z)-g_n(y)\|\geq N\epsilon\right).
\end{equation}
The terms can be estimated individually:
\[
    \sup_{z\in B_{\delta}(y)}\|g_n(z)-g_n(y)\|\leq \begin{cases}
        0, &\text{if}\;\; \Delta_n\cap B_{\delta}(y) = \emptyset\\
        2C, &\text{if}\;\; \Delta_n\cap B_\delta(y)\neq\emptyset.
    \end{cases}
\]
The termwise estimate, when inserted back into~\cref{eq: termwise bound and sup}, gives 
\[
    \mathbb{P}\left(\sup_{z\in B_\delta(y)}\left\|\sum_{n=1}^Ng_n(z)-g_n(y)\right\|\geq N\epsilon\right) \leq \mathbb{P}\left(2C\sum_{n=1}^N \chi(\Delta_n\cap B_{\delta}(y))\geq N\epsilon\right),
\]    
where the set function $\chi(A)$ is $0$ if $A=\emptyset$ and $1$ otherwise. The Chebyshev inequality then results in%Then by the classic Chernoff bound for Bernoulli random variables (\textcolor{red}{cite})
\[
  %\mathbb{P}\left(\sum_{n=1}^N \chi(\Delta_n\cap B_\delta(y))\geq \frac{N\epsilon}{2C}\right)\leq   \exp\left(-N D_\mathrm{KL}\left(\mathbb{P}(\Delta_1\cap B_\delta(y)\neq \emptyset)\middle\|(2C)^{-1}\epsilon\right)\right),
  \mathbb{P}\left(\sum_{n=1}^N \chi(\Delta_n\cap B_\delta(y))\geq \frac{N\epsilon}{2C}\right)\leq  \frac{4C^2\mathbb{E}[\chi(\Delta_n\cap B_\delta(y)\neq\emptyset)^2]}{\epsilon^2N}.
\]
%The proof is concluded by estimating the variance of a Bernoulli distributed random variable $\chi\sim\mathrm{Be}(p)$ as $\mathrm{Var}(\chi)=p(1-p)\leq p$.
The proof is concluded since $\mathbb{E}[\chi^2]=p$ for any Bernoulli variable $\chi\sim\mathrm{Be}(p)$.
\end{proof}

The idea is to apply \cref{lemma: discontinuous sum} by setting $g_n$ and $\Delta_n$ to $\falign{\bullet}{\xi_n}$ and $\Delta_{\xi_n}=\cup_{i=1}^L\cup_{j\neq i}\{z\colon (\sigma_i\xi_n,z)=(\sigma_j\xi_n,z)\}$ (the discriminant for the random data $\xi_n$), respectively. The result does not immediately hold, since the terms $\falign{z}{\xi_n}$ are not uniformly bounded. We will avoid this issue by conditioning on a uniform bound on the noise. Such a bound is a rare event, and we can use a Chernoff bound to control the probability:
\begin{lemma}\label{lemma: chernoff bound for chi2}
    Let $(\Omega,\mathcal{F},\mathbb{P})$ be a probability space and $x\colon\Omega \to \mathbb{R}^L$ random vector whose elements are normal in distribution with mean zero and variance one. Then, for any $C>0$ the following holds:
    \[
        \mathbb{P}\left(\|x\|\geq C\right)\leq \sqrt{e/L}^LC^L\exp\left(-C^2/2\right).
    \]
\end{lemma}
\begin{proof}
    The proof is a standard Chernoff bound that uses the Markov inequality. For all $t>0$, we have
    \[
        \mathbb{P}\left(\|x\|\geq C\right) =\mathbb{P}\left(\exp\left(t\sum_{\ell=1}^Lx[\ell]^2\right)\geq \exp(tC^2)\right)\leq\mathbb{E}[\exp(tx[1]^2)]^Le^{-tC^2}
=M(t)^L e^{-tC^2},
\]
where $M(t)=1/\sqrt{1-2t}$ is the moment-generating function of a $\chi^2(1)$ random variable. Optimization over $0\leq t<\tfrac{1}{2}$ by setting the derivative of $M(t)^Le^{-tC^2}$ to zero results in $\epsilon^2(1-2t)-L=0$ which has the solution $t=\tfrac{1}{2}-\tfrac{L}{2C^2}$, which inserted into the upper bound gives
\[
    \mathbb{P}(\|x\|\geq C) \leq \tfrac{C^L}{\sqrt{L}^L}\exp\left(-C^2/2+L/2\right) = (\sqrt{e/L})^L C^L\exp\left(-C^2/2\right).
\]
\end{proof}
We will also need a bound for a truncated normal distribution. In the following, we will denote by $\mathcal{N}(\mu,\Sigma\mid C)$ a truncated multivariate normal distribution, by which we mean that $x\sim \mathcal{N}(\mu,\Sigma\mid C)$ is defined by $x=y\mid \|y\|\leq C$ with $y\sim\mathcal{N}(\mu,\Sigma)$.
\begin{lemma}\label{result: gaussian bound}
     Let $x\sim\mathcal{N}(0,\Sigma\mid C)$ be a truncated multivariate normal with mean $\mu\in\mathbb{R}^L$, covariance matrix $\Sigma\in\mathbb{R}^{L\times L}$ and bound $C>0$. Then, for any $v\in\mathbb{R}^L$ and numbers $a,b\in\mathbb{R}$ such that $a\leq b$, the following holds:
     \[
        \mathbb{P}(v^Tx\in (a,b))\leq \frac{\Phi(b/\|v\|_\Sigma)-\Phi(a/\|v\|_\Sigma)}{\mathbb{P}(\|y\|\leq C)},
     \]
     where $\Phi$ is the cumulative distribution function for a standard Gaussian, $y\sim\mathcal{N}(\mu,\Sigma)$ is the non-truncated version of $x$, and $\|v\|_\Sigma = v^T\Sigma v$.
\end{lemma}
\begin{proof}
    For any events $A,B \subset \Omega$, we have $\mathbb{P}(A\mid B)=\mathbb{P}(A\cap B)/\mathbb{P}(B)\leq\mathbb{P}(A)/\mathbb{P}(B)$. Combining this inequality with $x=y\mid \|y\|\leq C$ (equality in distribution), we obtain
    \[
        \mathbb{P}(v^Tx\in (a,b)) = \mathbb{P}(v^Ty\in (a,b)\mid \|y\|\leq C)\leq \frac{\mathbb{P}(v^Ty\in(a,b))}{\mathbb{P}(\|y\|\leq C)}.
    \]
    Lastly, $v^Ty$ is a 1D Gaussian with mean zero and variance $\mathbb{E}[v^Tyy^Tv]=v^T\Sigma v=\|v\|_\Sigma$. This finishes the proof.
\end{proof}

Next, we use this last result to estimate the probability of $\mathbb{P}(\Delta_n\cap B_\delta(y)\neq\emptyset)$, that is discontinuities in the ball $B_\delta(y)$. In the following lemma, we prove that specifically for $\Delta_{\xi_1}$, there is an upper bound:
\begin{lemma}\label{lemma: discriminant bound}
    Let $(\Omega, \mathcal{F}, \mathbb{P})$ be a probability space and let $\xi_1=x+\epsilon_1$ where $x\in \mathbb{R}^L$ and $\epsilon_1\sim\mathcal{N}(0,\tau^2 I\mid C)$, with $I$ being the $L\times L$ identity matrix. Let $\Delta_{\xi_1}^{i,j}=\{z\colon(z,\sigma_i\xi_1)=(z,\sigma_j\xi_1)\}$ for $i,j\in\{1,2,\dots,L\}$ and $j\neq i$. Let $\Delta_{\xi_1}=\cup_{i\neq j}\Delta_{\xi_1}^{i,j}$. Then, for any non-periodic signal $y\in\mathbb{R}^L$ and sufficiently small $\delta>0$, it holds that
    \begin{equation}\label{eq: bound for discriminant intersect}
        \mathbb{P}(\Delta_{\xi_1}\cap B_\delta(y)\neq \emptyset)\leq \frac{4L^2\delta(C+\|x\|)}{\tau\sqrt{2\pi}\min_{i\neq j}\|\sigma_iy-\sigma_jy\|\mathbb{P}(\|r_1\|\leq C)},
    \end{equation}
    where $r_1\sim\mathcal{N}(0,\tau^2 I)$ is the non-truncated version of $\epsilon_1$. 
\end{lemma}
\begin{proof}
    Since $\Delta_{\xi_1}=\cup_{i\neq j}\Delta_{\xi_1}^{i,j}$, it holds that
\[
    \mathbb{P}(\Delta_{\xi_1}\cap B_\delta(y)\neq \emptyset)\leq \sum_{i\neq j}\mathbb{P}(\Delta_{\xi_1}^{i,j}\cap B_\delta(y)\neq \emptyset) \leq L^2 \max_{i\neq j}\mathbb{P}(\Delta_{\xi_1}^{i,j}\cap B_\delta(y)\neq\emptyset).
\]
We let $\nu^{i,j} = \sigma_i\xi_1-\sigma_j\xi_1$ denote the normal vector (not normalized) of the hyperplane $\Delta_{\xi_1}^{i,j}$, and let $\phi^{i,j} = \sigma_{-i}y-\sigma_{-j}y$. Then, $\Delta_{\xi_1}^{i,j}\cap B_\delta(y)\neq \emptyset$ exactly when the distance from $y$ to the hyperplane $\Delta_{\xi_1}^{i,j}$ is less than $\delta$. Since $\Delta_{\xi_1}^{i,j}$ intersects the origin, the distance from $y$ to $\Delta_{\xi_1}^{i,j}$ is simply the absolute value of the inner product between $y$ and the normal vector $\nu^{i,j}$. 
\begin{align*}
    \mathbb{P}(\Delta_{\xi_1}^{i,j}\cap B_\delta(y)\neq \emptyset)&=\mathbb{P}(-\delta\|\nu^{i,j}\|\leq (y,\nu^{i,j})\leq \delta\|\nu^{i,j}\|)\\
    &=\mathbb{P}\left(-\delta\|\nu^{i,j}\|-(\phi^{i,j},x)\leq (\phi^{i,j},\epsilon_1)\leq \delta\|\nu^{i,j}\|-(\phi^{i,j},x)\right)
\end{align*}
Where we used $(\sigma_ia,b)=(a,\sigma_{-i}b)$ for any vectors $a,b$ and cyclic shift $\sigma_i$. Since $\epsilon_1$ is a truncated Gaussian, $\|\xi_1\|\leq \|x\|+C$ and hence, $\|\nu^{i,j}\|\leq 2(\|x\|+C)$. By \cref{result: gaussian bound},
\begin{align*}
    \mathbb{P}&(\Delta_{\xi_1}^{i,j}\cap B_\delta(y)\neq \emptyset) \\
    &= \mathbb{P}(-2\delta(C+\|x\|)-(\phi^{i,j},x)\leq (\phi^{i,j},\epsilon_1)\leq 2\delta(C+\|x\|)-(\phi^{i,j},x))\\
    &\leq\frac{\Phi\left(\frac{2\delta(C+\|x\|)-(\phi^{i,j},x)}{\tau\|\phi^{i,j}\|}\right)-\Phi\left(\frac{-2\delta(C+\|x\|)-(\phi^{i,j},x)}{\tau\|\phi^{i,j}\|}\right)}{\mathbb{P}(\|r_1\|\leq C)}
\end{align*}
We now use the mean value theorem together with the integral form $\Phi(b)-\Phi(a)=\int_a^b\Phi'(x)\mathrm{d}x$ to write $\Phi(b)-\Phi(a)= (b-a)\Phi'(c)$ for some $c\in (a,b)$, and then bound the Gaussian density function $\Phi'(c)$ by its global maximum value $1/\sqrt{2\pi}$ and maximize over $i\neq j\in\{1,\dots,L\}$, to obtain
\[
    \mathbb{P}(\Delta_{\xi_1}\cap B_\delta(y)\neq \emptyset)
    \leq \frac{4L^2\delta(C+\|x\|)}{\sqrt{2\pi}\tau\min_{i\neq j}\|\phi^{i,j}\|\mathbb{P}(\|r_1\|\leq C)}.
\]
\end{proof}
We are finally ready to estimate $r_N(y,\epsilon,\delta)$. The following proposition summarizes our results.
\begin{proposition}\label{prop: rN bound}
    Let $(\Omega,\mathcal{F},\mathbb{P})$, $\mathcal{X}_N=(\xi_n)_{n=1}^N$ with $\xi_n=x+\epsilon_n$ and $f_N, f$ be as in \cref{setting: alignment} with $r_N$ as in~\cref{eq:rN}. Then, the following bound holds for any $\epsilon>0,\delta>0$ and non-periodic $y\in\mathbb{R}^L$:
    \begin{equation}
        r_N(y,\epsilon,\delta)\leq  \frac{\sqrt{eL^{-L}\log N}^L}{N}+\frac{64L^2\tau}{\sqrt{2\pi}\min_{i,j}\|\sigma_{-j}y-\sigma_{-i}y\|}\frac{\delta\log N}{N\epsilon^2}.
    \end{equation}
\end{proposition}
\begin{proof}
     For any events $A,B$, it holds that
     \begin{equation}\label{eq: set inequality}
     \mathbb{P}(A)\leq \mathbb{P}(B^c \cup (A\cap B))\leq \mathbb{P}(B^c)+\mathbb{P}(A\mid B)\mathbb{P}(B)
     \end{equation}
     Let $C>0$ be a positive number, which will later be chosen appropriately large. We will also use that $\mathbb{P}(\max_n\|\epsilon_n\|\geq C)=\mathbb{P}(\cup_n\{\|\epsilon_n\|\geq C\})\leq N\mathbb{P}(\|\epsilon_1\|\geq C)$ for any such $C$. Let $B=\{\max_n\|\epsilon_n\|\leq C\}$ and $A=\{\sup_{z\in B_\delta(y)}\|f_N(z)-f_N(y)\|\geq \epsilon\}$. Then~\cref{eq: set inequality} gives:
    \begin{align*}
        r_N(y&,\epsilon,\delta) = \mathbb{P}\left(\sup_{z\in B_\delta(y)}\|f_N(z)-f_N(y)\|\geq \epsilon\right)\\
        &\leq \mathbb{P}(\max_n\|\epsilon_n\|\geq C)+\mathbb{P}\left(\sup_{z\in B_\delta(y)}\|f_N(z)-f_N(y)\|\geq \epsilon\;\middle|\;\max_n\|\epsilon_n\|\leq C\right)\mathbb{P}(\max_n\|\epsilon_n\|\leq C)\\
        & \leq N\mathbb{P}(\|\epsilon_1\|\geq C)+\frac{2C\mathbb{P}(\Delta_{\xi_1}\cap B_\delta(y)\neq\emptyset\mid\|\epsilon_1\|\leq C)}{\epsilon^2N}\mathbb{P}(\|\epsilon_1\|\leq C)^{N},
    \end{align*}
    where we used~\cref{lemma: discontinuous sum} on the second term in the last step, since $f_N$ is a sum of piecewise constant, uniformly bounded random functions $\falign{\bullet}{\xi_n}$ conditioned on $\|\epsilon_n\|\leq C$. Next, we apply~\cref{lemma: discriminant bound}, which results in
    \begin{align*}
    q_N&(y,\epsilon,\delta)\\
    & \leq N\mathbb{P}(\|\epsilon_1\|\geq C)+\frac{8CL^2\delta(C+\|x\|)}{\tau\epsilon^2N\sqrt{2\pi}\min_{i,j}\|\sigma_{-j}y-\sigma_{-i}y\|\mathbb{P}(\|\epsilon_1\|\leq C)}\mathbb{P}(\|\epsilon_1\|\leq C)^{N},\\
    & \leq N\mathbb{P}(\|\epsilon_1\|\geq C)+\frac{8L^2(1+\|x\|/C)}{\tau\sqrt{2\pi}\min_{i,j}\|\sigma_{-j}y-\sigma_{-i}y\|}\frac{C^2\delta}{\epsilon^2N},
    \end{align*}
    where we used $\mathbb{P}(\|\epsilon_1\|\leq C)\leq 1$ to eliminate $\mathbb{P}(\|\epsilon_1\|\leq C)^{N-1}$ from the second term of the inequality. We now use~\cref{lemma: chernoff bound for chi2}, which shows
    \[
        r_N(y,\epsilon,\delta)\leq N\sqrt{e/L\tau^2}^LC^L\exp(-C^2/2\tau^2)+\frac{8L^2(1+\|x\|/C)}{\tau\sqrt{2\pi}\min_{i\neq j}\|\sigma_{-j}y-\sigma_{-i}y\|}\frac{C^2\delta}{N\epsilon^2}.
    \]
    Let $C=2\tau\sqrt{\log N}$ so that $C^2/(2\tau^2)=2\log N$. Then, for $N$ sufficiently large, say $N>M$ with $M$ independent of $\epsilon, \delta$, it holds that $\|x\|/C\leq 1$. Hence, 
    \begin{align*}
        r_N(y,\epsilon,\delta)&\leq  \frac{\sqrt{4eL^{-1}\log N}^L}{N}+\frac{64L^2\tau}{\sqrt{2\pi}\min_{i\neq j}\|\sigma_{-j}y-\sigma_{-i}y\|}\frac{\delta\log N}{N\epsilon^2}.
    \end{align*}
\end{proof}

\subsection{Combining the bounds} Combining~\cref{lemma: random bound}, \cref{prop: rN bound} and \cref{eq: chebyshev qn bound} gives the proof of~\cref{thm: iterated convergence}.
\begin{proof}[Proof of \cref{thm: iterated convergence}]
    We first use~\cref{prop: rN bound} and~\cref{eq: chebyshev qn bound} to bound equation~\cref{eq: inequality recursion}, for the special case when $\delta=\rho\epsilon$ with $\rho=\min\{0.99,\tfrac{1}{3C_a}\}<1$, we obtain:
    \begin{align*}
         &p_{N,k}(\epsilon) \leq p_{N,k-1}(\rho\epsilon) + q_N(f^{k-1}(y),\tfrac{\epsilon}{3C_a})+r_N(f^{k-1}(y),\tfrac{\epsilon}{3 C_a}, \rho\epsilon)\\
         &\leq p_{N,k-1}(\rho\epsilon) + \frac{\mathrm{Var}(\falign{f^{k-1}(y)}{\xi_1})}{(3C_a)^{-2}\epsilon^2 N} + \frac{\sqrt{4e\log N}^L}{L^{L/2}N}+\frac{576 \rho C_a L^2\tau\sqrt{2\pi}^{-1}(N\epsilon)^{-1}\log N}{\min_{i\neq j}\|\sigma_{-j}f^{k-1}(y)-\sigma_{-i}f^{k-1}(y)\|}\\
         &\leq p_{N,k-1}(\rho\epsilon) +\frac{A_1(\log N)^{L/2}}{N} + \frac{A_2(\log N)}{\epsilon N} + \frac{A_3}{\epsilon^2 N},
    \end{align*}
    Where we defined the constants $A_1 = \sqrt{4eL^{-1}}^L$, $A_2=\max_{\ell\leq k}192C_a L^2\tau\sqrt{2\pi}^{-1}(\min_{i\neq j}\|\sigma_{-i}f^{\ell-1}(y)-\sigma_{-j}f^{\ell-1}(y)\|)^{-1}$, and $A_3=\max_{\ell\leq k}9C_a^2\mathrm{Var}(\falign{f^{\ell-1}(y)}{\xi_1})$. We expand the recursion:
    \begin{align*}
        p_{N,k}(\epsilon) &\leq \frac{A_1(\log N)^{L/2}}{N}\sum_{\ell=0}^{k-1} 1  + \frac{A_2(\log N)}{\epsilon N}\sum_{\ell=0}^{k-1}\rho^{-\ell} + \frac{A_3}{\epsilon^2 N}\sum_{\ell=0}^{k-1}\rho^{-2\ell}\\
        &=k\frac{A_1(\log N)^{L/2}}{N} + \left(\frac{\rho^{-k}-1}{\rho^{-1}-1}\right)\frac{A_2(\log N)}{\epsilon N} + \left(\frac{\rho^{-2k}-1}{\rho^{-2}-1}\right)\frac{A_3}{\epsilon^2 N}\\
        &\leq k A_1\frac{(\log N)^{L/2}}{N} + \left(\left(\frac{\rho^{-k}-1}{\rho^{-1}-1}\right)A_2+\left(\frac{\rho^{-2k}-1}{\rho^{-2}-1}\right)\frac{A_3}{\log N}\right)\frac{\log N}{\epsilon^2N}        
    \end{align*}
    where we used $p_{N,0}(\epsilon)=0$ and $\epsilon\leq1$. Lastly, for  $N$ large enough, the $A_2$-contribution dominates the last term and hence, letting $C_1=kA_1$ and $C_2=2(\rho^{-k}-1)/(\rho^{-1}-1)A_2$ gives the desired bound.
 \end{proof}

\end{document}